\documentclass{article}
\usepackage[utf8]{inputenc}
\usepackage{microtype}%if unwanted, comment out or use option "draft"
\usepackage{xspace}
\usepackage{graphicx}
\usepackage{amsmath}
\usepackage{amsfonts}
\usepackage{amsthm}
\usepackage{xspace}
\usepackage{tcolorbox}
\usepackage{fullpage}
\theoremstyle{plain}

\newcommand{\lfin}{$l_1 \in [k_1]_0$}
\newcommand{\lsin}{$l_2 \in [k_2]_0$}

\newcommand{\lstarp}{$L^{*}_p$}
\newcommand{\dstar}{${D}^{*}$}

\newcommand{\fptfull}{\textsc{Fixed-Parameter Tractable}\xspace}
\newcommand{\fpt}{\textsc{FPT}\xspace}
\newcommand{\npc}{\textsc{NP}-complete\xspace}

\newcommand{\abjpfull}{\textsc{Annotated Bipartite-BJB}\xspace}
\newcommand{\abjp}{\textsc{AB-BJB}\xspace}

\newcommand{\abcjpfull}{\textsc{Annotated Bipartite Connected-BJB}\xspace}
\newcommand{\abcjp}{\textsc{ABC-BJB}\xspace}

\newcommand{\jpfull}{\textsc{Balanced Judicious Bipartition}\xspace}
\newcommand{\jp}{\textsc{BJB}\xspace}

\newcommand{\hpfull}{\textsc{Hypergraph Painting}\xspace}
\newcommand{\hp}{\textsc{HP}\xspace}

\newcommand{\octfull}{\textsc{Odd Cycle Transversal}\xspace}
\newcommand{\oct}{\textsc{OCT}\xspace}

\newcommand{\OO}{\mathcal{O}\xspace}

\newcommand{\defparproblem}[4]{
  \vspace{3mm}
\noindent\fbox{
  \begin{minipage}{.97\textwidth}
  \begin{tabular*}{\textwidth}{@{\extracolsep{\fill}}lr} \textsc{#1}  & {\bf{Parameter:}} #3 \\ \end{tabular*}
  {\bf{Input:}} #2  \\
  {\bf{Question:}} #4
  \end{minipage}
  }
  \vspace{2mm}
}

\newcommand{\defparoutproblem}[4]{
  \vspace{3mm}
\noindent\fbox{
  \begin{minipage}{.97\textwidth}
  \begin{tabular*}{\textwidth}{@{\extracolsep{\fill}}lr} \textsc{#1}  & {\bf{Parameter:}} #3 \\ \end{tabular*}
  {\bf{Input:}} #2  \\
  {\bf{Output:}} #4
  \end{minipage}
  }
  \vspace{2mm}
}

\newcommand{\defproblem}[3]{
  \vspace{3mm}
\noindent\fbox{
  \begin{minipage}{.97\textwidth}
  \begin{tabular*}{\textwidth}{@{\extracolsep{\fill}}lr} #1  \\ \end{tabular*}
  {\bf{Input:}} #2  \\
  {\bf{Output:}} #3
  \end{minipage}
  }
  \vspace{2mm}
  }

\newtheorem{theorem}{Theorem}
\newtheorem{definition}{Definition}
\newtheorem{lemma}{Lemma}  
\newtheorem{claim}{Claim}
\newtheorem{proposition}{Proposition}

\newcommand{\yes}{\textsc{YES}}
\newcommand{\no}{\textsc{NO}}

\title{\jpfull{} is \fptfull}

\author{Daniel Lokshtanov\footnote{University of Bergen, Norway, daniello@ii.uib.no.} , Saket Saurabh\footnote{Institute of Mathematical Sciences, HBNI, Chennai, India, UMI Relax and University of Bergen, Norway, saket@imsc.res.in.} , Roohani Sharma\footnote{Institute of Mathematical Sciences, HBNI, Chennai, India and UMI ReLax, roohani@imsc.res.in.} , Meirav Zehavi\footnote{Ben-Gurion University, Beersheba, Israel, meiravze@bgu.ac.il.}}
\begin{document}

\maketitle

\begin{abstract}
The family of judicious partitioning problems, introduced by Bollob{\'a}s and Scott to the field of extremal combinatorics, has been extensively studied from a structural point of view for over two decades. This rich realm of problems aims to counterbalance the objectives of classical partitioning problems such as {\sc Min Cut}, {\sc Min Bisection} and {\sc Max Cut}. While these classical problems focus solely on the minimization/maximization of the number of edges crossing the cut, judicious (bi)partitioning problems ask the natural question of the minimization/maximization of the number of edges lying in the (two) sides of the cut. In particular,  {\sc Judicious Bipartition (JB)} seeks a bipartition that is ``judicious'' in the sense that neither side is burdened by too many edges, and {\sc Balanced JB} also requires that the sizes of the sides themselves are ``balanced'' in the sense that neither of them is too large. Both of these problems were defined in the work by Bollob{\'a}s and Scott, and have received notable scientific attention since then. In this paper, we shed light on the study of judicious partitioning problems from the viewpoint of algorithm design. Specifically, we prove that {\sc BJB} is FPT (which also proves that {\sc JB} is FPT). 
\end{abstract}

\section{Introduction}
More than twenty years ago, Bollob{\'a}s and Scott \cite{bollobas1993judicious} defined the notion of \emph{judicious partitioning problems}. Since then, the family of judicious partitioning problems has been extensively studied in the field of Extremal Combinatorics, as can be evidenced by the abundance of structural results described in surveys such 
as~\cite{bollobas2002problems,scott2005judicious}. This rich realm of problems aims to counterbalance the objectives of classical partitioning problems such as {\sc Min Cut}, {\sc Min Bisection}, {\sc Max Cut} and  {\sc Max Bisection}. While these classical problems focus solely on the minimization/maximization of the number of edges crossing the cut, judicious (bi)partitioning problems ask the natural questions of the minimization/maximization of the number of edges lying in the (two) sides of the cut. Another significant feature of judicious partitioning problems that also distinguishes them from other classical partitioning problems is that they inherently and naturally encompass several objectives, aiming to minimize (or maximize) the number of edges in several sets {\em simultaneously}.

In this paper, we shed light on properties of judicious partitioning problems from the viewpoint of the design of algorithms. Up until now, the study of such problems has essentially been overlooked at the algorithmic front, where one of the underlying reasons for this discrepancy might be that standard machinery does not seem to handle them effectively. Specifically, we focus on the {\sc Judicious Bipartition} problem, where we seek a bipartition that is ``judicious'' in the sense that neither side is burdened by too many edges, and on the {\sc Balanced Judicious Bipartition} problem, where we also require that the sizes of the sides themselves are ``balanced'' in the sense that neither of them is too large. Both of these problems were defined in the work by Bollob{\'a}s and Scott, and have received notable scientific attention since then. Formally, {\sc Balanced Judicious Partition} is defined as~follows.

\defparproblem{\jpfull{} (\jp)}{A multi-graph $G$, and  integers $\mu$, $k_1$ and $k_2$}{$k_1 + k_2$}{Does there exist a partition $(V_1,V_2)$ of $V(G)$ such that $|V_1| = \mu$ and for all $i \in \{1,2\}$, it holds that $|E(G[V_i])| \leq k_i$?}

We note  that in the literature, the term \jp\ refers to the case where $\mu=\lceil \frac{|V(G)|}{2} \rceil$, and hence it is more restricted then the definition above. By dropping the requirement that $|V_1| = \mu$, we get the {\sc Judicious Bipartition (JB)} problem. 
By using new crucial insights into these problems on top of the most advanced machinery in Parameterized Complexity to handle partitioning problems,\footnote{To the best of our knowledge, up until now, this machinery has actually only been proven useful to solve one natural problem which could not have been tackled using earlier tools.} we are able to resolve the question of the Parameterized Complexity of {\sc BJB} (and hence also of {\sc JB}). In particular, we prove the following theorem. 

\begin{theorem}\label{thm:jp}
\jp{} can be solved in time $2^{k^{\OO(1)}} \cdot |V(G)|^{\OO(1)}$.
\end{theorem}

\noindent
{\bf Structural Results.} Denote $n=|V(G)|$ and $m=|E(G)|$. To survey several structural results about judicious partitioning problems, we first define the notions of {\em $t$-cut} and 
{\em max (min) $t$-judicious partitioning}. Given a partition of  $V(G)$ into $t$ parts, a $t$-cut is the number of edges going across the parts, while a max (min) judicious $t$-partitioning is the maximum  (minimum) number of edges in any of the parts. When $t=2$, we use the standard terms {\em bipartite-cut} and  {\em judicious bipartitioning}, respectively. Furthermore, %unless explicitly stated otherwise, 
 by $t$-judicious partitioning we mean max $t$-judicious partitioning. As stated earlier,  Bollob{\'a}s and Scott \cite{bollobas1993judicious} defined the notion of \emph{judicious partitioning problems} in 1993. In that paper, they showed that for any positive integer $t$ and graph $G$, we can partition  $V(G)$ into $t$ sets, $V_1, \ldots, V_t$, so that  $|E(G[V_i])| \leq \frac{t}{t+1}m$ for all $i\in \{1,\ldots, t\}$. 
Bollob{\'a}s and Scott also studied this problem on 
graphs of maximum degree $\Delta$, and showed that there exists a partition of $V(G)$ into $t$ sets $V_1, \ldots, V_t$  
so that it simultaneously satisfies  an upper bound and a lower bound on the number of edges in each part as well as on edges between every pair of parts. 
Later,  Bollob{\'a}s and Scott~\cite{bollobas2002problems} gave several new results, leaving open other new questions around judicious partitioning.  In \cite{bollobas2004judicious} they showed an optimal bound for judicious partitioning on bounded-degree graphs. 
These problems have also been studied on general hypergraphs \cite{bollobas1997judicious}, uniform hypergraphs~\cite{haslegrave2014judicious}, $3$-uniform hypergraphs~\cite{bollobas2000judicious} and directed graphs \cite{lee2016judicious}.

The special cases of judicious partitioning problems called judicious bipartitioning and balanced judicious  bipartitioning problems have also been studied intensively.  Bollob{\'a}s and Scott~\cite{bollobas1999exact} proved an upper 
bound on judicious bipartitioning and proved that every graph that achieves the essentially best known 
lower bound on bipartite-cut, given by Edwards in \cite{edwards1973some} and \cite{edwards1975improved}, also achieves this upper bound for  
judicious bipartitioning. In fact, they showed that this is exact for complete graphs of odd order, which are the only extremal graphs without isolated vertices. Alon et al.~\cite{alon2003maximum} gave a non-trivial connection between the size of a bipartite-cut in a graph and judicious partitioning into two sets. In particular, they showed that if a graph has a bipartite-cut of size  at least $\frac{m}{2}+\delta$ where $\delta\leq m/30$, then there exists a bipartition $(V_1,V_2)$ of $V(G)$ such that $|E(G[V_i])| \leq \frac{m}{4}-\frac{\delta}{2} + \frac{10 \delta^2}{m}+3 \sqrt{m}$ for $i\in \{1,2\}$. They complemented these results by showing an upper bound on the number of edges in each part  when $\delta> m/30$. %Finally, they related these results to the girth of a graph.
  Bollob{\'a}s and Scott~\cite{bollobas2010max} studied similar relations between 
 $t$-cuts and $t$-judicious partitionings for $t\geq 3$. Recently, these results were further refined~\cite{xu2009judicious,ma2016judicious}. Xu et al.~\cite{XuYY10} and Xu and Yu~\cite{XuY14} studied balanced 
judicious bipartitioning  where both parts are of almost equal size (that is, one of the sizes is $\lceil \frac{n}{2} \rceil$). Both of these papers concern the following conjecture of Bollob{\'a}s and Scott~\cite{bollobas2002problems}:   if $G$ is a graph with minimum degree of at least $2$, then $V(G)$ admits a balanced bipartition $(V_1,V_2)$ such that for each $i\in\{1,2\}$, 
$|E(G[V_i])| \leq \frac{m}{3}$.  % Xu and Yu~\cite{XuY14} confirmed this conjecture.
For further results on judicious partitioning, we refer to the surveys~\cite{bollobas2002problems,scott2005judicious}. 

\smallskip
\noindent 
{\bf Algorithmic Results.} While classical partitioning problems such as {\sc Min Cut}, {\sc Min Bisection}, {\sc Max Cut} and  {\sc Max Bisection} have been studied extensively algorithmically, the same is not true about judicious partitioning problems. Apart from  {\sc Min Cut}, all the above mentioned partitioning problems are \npc.  These \npc partitioning problems were investigated by all algorithmic paradigms meant for coping with \npc, including approximation algorithms and parameterized complexity. In what follows, we discuss known results related to these problems in the realm of parameterized complexity.

First, note that for every graph $G$, there always exists a bipartition of the vertex set  into two parts (in fact equal parts~\cite[Corollory 1]{gutin2010note}) such that at least $m/2$ edges are going across. This immediately implies that 
{\sc Max Cut} and {\sc Max Bisection} are \fpt when parameterized by the cut size (the number of edges going across the partition). This led Mahajan and Raman \cite{mahajan1999parameterizing} to introduce the notion of above-guarantee parameterization.  In particular, they showed that one can decide whether a graph has a bipartite-cut of size $\frac{m}{2}+k$ in time $\OO(m+n+ k4^k)$.  However, 
Edwards~\cite{edwards1973some} showed that every connected graph $G$ has a bipartite-cut of size $\frac{m}{2}+\frac{n-1}{4}$. Thus, a more interesting question asks whether finding a bipartite-cut of size at least $\frac{m}{2}+\frac{n-1}{4}+k$ is \fpt. Crowston et al.~\cite{CrowstonJM15} showed that indeed this is the case as they design an algorithm with running time $\OO(8^kn^4)$. Recently, Etscheid and Mnich \cite{DBLP:conf/isaac/EtscheidM16} discovered a kernel with a linear number of vertices (improving upon a kernel by Crowston et al.~\cite{CrowstonGJM13}), and the aforementioned algorithm was sped-up to run in time $\OO(8^km)$ \cite{DBLP:conf/isaac/EtscheidM16}.
 Gutin and Yeo studied an above-guarantee version of  {\sc Max Bisection}  \cite{gutin2010note}, proving that finding a balanced bipartition such that it has at least $\frac{m}{2}+k$ edges is 
\fpt (also see~\cite{MnichZ12}).\footnote{We refer to surveys~\cite{MahajanRS09,GutinY12} for details regarding  above-guarantee parameterizations.} In this context {\sc Max Bisection}, it is also relevant to mention the {\sc $(k,n-k)$-Max Cut}, which asks for a bipartite-cut of size at least $p$ where one of the sides is of size exactly $k$. Parameterized by $k$, this problems is W[1]-hard \cite{DBLP:journals/cj/Cai08}, but parameterized by $p$, this problem is solvable in time $\OO^*(2^p)$ \cite{DBLP:conf/latin/SaurabhZ16} (this result improved upon algorithms given in \cite{DBLP:journals/algorithmica/BonnetEPT15,DBLP:journals/jcss/ShachnaiZ16}).

Until recently, the parameterized complexity of  {\sc Min Bisection} was open. Approaches to tackle this problem  materialized when the parameterized complexity of {\sc $\ell$-Way Cut} was resolved. Here, given a graph $G$ and positive integers $k$ and $\ell$, the objective is to delete at most $k$ edges from $G$ such that  it has at least $\ell$ components.  Kawarabayashi and  Thorup~\cite{KawarabayashiT11} showed that this problem is \fpt. Later, Chitnis et al.~\cite{chitnis2012designing} developed a completely new tool based on this, called {\em randomized contractions}, to deal with plethora of cut problems. Other cut problems that have been shown to be \fpt include the generalization of {\sc Min Cut} to {\sc Multiway Cut} and {\sc Multicut}~\cite{ChenLL09,Marx06,MarxR14}. Eventually, Cygan et al. \cite{cygan2014minimum}, combining ideas underlying the algorithms developed for {\sc Multiway Cut}, {\sc Multicut},  {\sc $\ell$-Way Cut} and randomized contractions together with a new kind of decomposition, showed {\sc Min Bisection} to be \fpt.
Finally, let us also mention the min $c$-judicious partitioning (which is a maximization problem), called {\sc $c$-Load Coloring}, where given a graph $G$ and a positive integer $k$, the goal is to decide whether $V(G)$ can be partitioned into $c$ parts so that  each part has at least $k$ edges. Barbero et al.~\cite{barbero2015parameterized} showed that this problem is \fpt (also see~\cite{gutin2014parameterized}).

Despite the abundance of work described above, the parameterized complexity of  {\sc JB}  and \jp has not yet considered. 
We fill this gap in our studies by showing that both of these problems  are \fpt. It is noteworthy to remark that one can show that the generalization of  {\sc Min Bisection} to {\sc $c$-Min Bisection}, where the objective is to find a partition into $c$-parts such that each part are almost equal and there are at most $k$ edges going across different parts, is \fpt~\cite{cygan2014minimum}. However, such a generlization is not possible for either {\sc JB}  or \jp. Indeed, even the existence of an algorithm with running time $n^{f(k)}$, for any arbitrary function $f$, would imply a polynomial-time algorithm for {\sc $3$-Coloring}, where $k$ is set to $0$.  

\smallskip
\noindent 
{\bf Our Approach.} For the sake of readability, our strategy of presentation of our proof consists of the definition of a series of problems, each more ``specialized'' (in some sense) than the previous one, where each section shows that to eventually solve \jp, it is sufficient to focus on some such problem rather than the previous one. We start by showing that we can focus on the solution of the case of \jp{} where the input graph is bipartite at the cost of the addition of annotations. For this purpose, we present a (not complicated) Turing reduction that employes a known algorithm for the \octfull{} problem (see Section \ref{sec:jp}). The usefulness of the ability to assume that the input graph is bipartite is a key insight in our approach. In particular, the technical parts of our proof crucially rely on the observation that a connected bipartite graph has only two bipartitions (here, we consider bipartitions as ordered pairs). Keeping this intuition in mind, our next step is to reduce the current annotated problem to one where the input graph is also assumed to be connected (this specific argument relies on a simple application of dynamic programming).

Having at hand an (annotated) problem where the input graph is assumed to be a connected bipartite graph, we proceed to the technical part of our proof, which employs the (heavy) machinery developed by Cygan et al.~\cite{cygan2014minimum}. While this machinery primarily aims to tackle problems where one seeks small cuts in addition to some size constraint, our problem involves  a priori seemingly different type of constraints. Nevertheless, we observe that once we handle a connected graph, the removal of any set of $k$ edges (to deal with the size constraint and annotations) would not break the graph to more than $k+1$ connected components, and each of these components would clearly be a bipartite graph. Hence, we can view (in some sense) our problem as a cut problem. In practice, the relation between our problem and a cut problem is quite more intricate, and to realize our idea, we crucially rely on the fact that the connected components are bipartite graphs, which allows us to ``guess'' a binary vector specifying the biparition of their vertex sets in the final solution. This operation entitles the employment of coloring functions (employing $k+1$ colors) and their translation into bipartitions (which at a certain point in our paper, we would start viewing as colorings employing two colors). Let us remark that the machinery introduced by~\cite{cygan2014minimum} is the computation of a special type of tree decomposition. Accordingly, our approach would eventually involve the introduction of a specialization of \jp{} that aims to capture the work to perform when handling a bag of the tree decomposition. The definition of this specific problem is very technical, and hence we defer the description of related intuitive explanations to the appropriate locations in Section \ref{sec:abcbjb}, where we have already set up the required notations  to discuss it.

%Note that the proofs of statements marked by $(\star)$ are omitted due to space constraints and can be found in the full version of the paper \cite{lokshtanov2017judiciousfull}.

\section{Preliminaries}\label{sec:prelims}
\noindent 
{\bf General Notation.} 
Let $f : A \to B$ be some function. Given $A' \subseteq A$, the notation $f(A') =b$ indicates that for all $a \in A'$, it holds that $f(a)=b$. 
An {\em extension} $f'$ of the function $f$ is a function whose domain $A'$ is a superset of $A$ and whose range is $B$, such that for all $a\in A$, it holds that $f'(a) = f(a)$. 
For any $A' \subseteq A$, the {\em restriction} $f|_{{A'}}$ of $f$ is a function from $A'$ to $B$ such that for any $a \in A'$, 
$f|_{{A'}}(a) = f(a)$. Bold face lowercase letters are used to denote tuples (vectors). For any tuple $\mathbf{v}$, we let $\mathbf{v}[i]$ denote the $i$th coordinate of $\mathbf{v}$. 
Given some condition $\psi$, we define $[\psi] =1$ if $\psi$ is true and $[\psi] =0$ otherwise.
For any positive integer $x$, we denote by $[x]$ the set $\{1,2, \ldots,x\}$ and by $[x]_0$ the set $\{0, 1, \ldots,x\}$.

\smallskip
\noindent 
{\bf Graph Theory.} 
Given a graph $G$, we let $V(G)$ and $E(G)$ denote the vertex-set and the edge-set of $G$, respectively. 
For a subset $A \subseteq V(G)$, we denote by $\delta(A)$ the set of boundary vertices of $A$, that is, $\delta(A) = \{v\in A :$ there exists $u \in V(G) \setminus A \text{ such that } \{u,v\} \in E(G)\}$.  
We let $G\setminus A$ denote the subgraph of $G$ induced by $V(G)\setminus A$.
A {\em bipartite graph} is a graph $G$ such that there exists a bipartition $(X,Y)$ of $V(G)$ where $X$ and $Y$ are independent sets. In this paper, we treat such bipartitions as {\em ordered pairs}. That is, if $(X,Y)$ is a bipartition of some bipartite graph $G$, then $(Y,X)$ is assumed to be a {\em different} bipartition of the graph $G$. 
 For {\em connected} bipartite graphs, we have the following simple yet powerful~insight.

\begin{proposition}[Folklore]\label{proposition:bipartite}
Any connected bipartite graph $G$ has exactly $2$ bipartitions, $(X,Y)$ and $(Y,X)$.
\end{proposition}

The {\em treewidth} of a graph aims to measure how close the graph is to a tree. Formally, this notion is defined as follows.

\begin{definition}
A {\em tree decomposition} of a graph $G$ is a pair $(T, \beta)$ such that $T$ is a rooted tree, $\beta : V(T) \to 2^{V(G)}$, and the following conditions are satisfied.

\begin{enumerate}
\setlength{\itemsep}{-1pt}
 \item For all $\{u,v\} \in E(G)$, there exists $t \in V(T)$ such that $u,v \in \beta(t)$.
 \item For all $v \in V(G)$, the subgraph of $T$ induced by $X_v =\{t: v \in \beta(t) \}$ is a (connected) subtree of $T$ on at least one node.
\end{enumerate}
\end{definition}

Given $t,\widehat{t}\in V(G)$, the notation $\widehat{t} \preceq t$ indicates that $\widehat{t}$ is a descendant of $t$ in $T$. Note that $t$ is a descendant of itself.
For any $t \in V(T)$, let $t'$ denote the unique parent of $t$ in $T$. We also need the standard notations $\sigma(t) = \beta(t) \cap \beta(t')$ and $\gamma(t) = \bigcup\limits_{\widehat{t} \preceq t } \beta(\widehat{t})$. 

\begin{proposition}[Folklore]
Let $(T, \beta)$ be a tree decomposition of a graph $G$. Given a node $t\in V(T)$, let $t_1, \ldots, t_s$ denote the children of $t$ in $T$, and for all $i\in[s]$, define $V_{t_i} = \gamma(t_i) \setminus \beta(t)$. 
Let $V_{t'} =V(G) \setminus (\beta(t) \cup \bigcup\limits_{i=1}^{s} V_{t_i})$.
Then, the vertex-set of each connected component of $G \setminus \beta(t)$ is a subset of one of $V_{t_1}, \ldots, V_{t_s}, V_{t'}$.
\end{proposition}

Let $H$ be some hypergraph. A s{\em panning forest} of $H$ is a subset $E'\subseteq E(H)$ of minimum size such that the set containing all endpoints of the hyperedges in $E'$ is equal to $V(H)$. In this paper, we implicitly assume that hypergraphs contain no isolated vertices.

\smallskip
\noindent 
{\bf Unbreakability.}  
A {\em separation} of a graph $G$ is a pair $(X,Y)$ such that $X,Y \subseteq V(G)$ and $X \cup Y = V(G)$. The {\em order of a separation} $(X,Y)$ is equal to $|X \cap Y|$.

\begin{definition}
Let $G$ be a graph, $A \subseteq V(G)$, and $q,k\in\mathbb{N}$. The set $A$ is said to be {\em $(q,k)$-unbreakable} in $G$ if for {\em every} separation $(X,Y)$ of $G$ of order at most $k$, either $|(X \setminus Y) \cap A| \leq q$ or $|(Y \setminus X) \cap A| \leq q$.
\end{definition}

We also define a notion of unbreakability in the context of functions.
\begin{definition}
A function $g : U \to [k]_0$ is called {\em $(q,k)$-unbreakable} if there exists $i \in [k]_0$ such that $\sum\limits_{j\in [k]_0 \setminus \{i\}} |g^{-1}(j)| \leq q$.
\end{definition}

Let us now claim that there do no exist ``too many'' $(q,k)$-unbreakable functions.
\begin{lemma}\label{lemma:numOfUnbreakableFunctions}
For all $q,k\in\mathbb{N}$, the number of $(q,k)$-unbreakable functions from a universe $U$ to $[k]_0$ is upper bounded by $ \sum\limits_{l=0}^{q} {|U|\choose l} \cdot q^k \cdot (k+1)$.
\end{lemma}
\begin{proof}
Let $g : U \to [k]_0$ be some $(q,k)$-unbreakable function. By the definition of a $(q,k)$-unbreakable function, there exists $i \in [k]_0$ such that $\sum\limits_{j\in [k]_0 \setminus i} |g^{-1}(j)| \leq q$. There are $(k+1)$ ways of choosing such an index $i$, $ \sum\limits_{l=0}^{q} {|U|\choose l}$ ways of choosing at most $q$ elements that are not mapped to $i$, and at most $q^k$ ways of partitioning this set of at most $q$ elements into $k$ parts. Thus, the total number of such functions $g$ is upper bounded by $ \sum\limits_{l=0}^{q} {|U|\choose l} q^k  (k+1)$.
\end{proof}

\section{Solving \jpfull}\label{sec:jp}
In this section, we prove Theorem \ref{thm:jp} under the assumption that we are given an algorithm for an annotated, yet restricted, variant of \jp{}. Throughout this section, an instance of \jp{} is denoted by $\jp(G,\mu,k_1,k_2)$, and we define $k = k_1 + k_2$. Given a partition $(V_1,V_2)$ that witnesses that an instance $\jp(G,\mu,k_1,k_2)$ is a \yes{}-instance, we think of the vertices in $V_1$ as colored 1 and the vertices in $V_2$ as colored 2; hence, we call such a partition a \emph{witnessing coloring} of $\jp(G,\mu,k_1,k_2)$. To prove Theorem \ref{thm:jp}, we first define the \octfull{} problem. Here, given a graph $G$, a set $S\subseteq V(G)$ is called an {\em odd cycle transversal} if $G\setminus S$ is a bipartite graph.

\defparproblem{\octfull{} (\oct)}{An undirected multi-graph graph $G$, and an integer $k$.}{$k$}{Does $G$ have an odd cycle transversal of size at most $k$?}

An instance of \octfull{} is denoted by $\oct(G,k)$. The algorithm given by the result below shall be a central component in  the design of our algorithm for~\jp.

\begin{proposition}[\cite{lokshtanov14}] \label{prop:oct}
\octfull{} can be solved in time $2.3146^k n^{\OO(1)}$.
\end{proposition}

Apart from \oct{}, we also need to define an auxiliary problem that we call \abjpfull\ (\abjp). As we proceed with our proofs, we shall continue defining auxiliary problems, where each problem captures a task more specific and technically more challenging than the previous one. The choice of this structure aims to ease the readability of our paper. 
 Intuitively, \abjp{} is basically the \jp{} problem on bipartite graphs, with an extra constraint that demands that certain vertices are assigned a particular color by the witnessing coloring. We remark that the necessity of the reduction to bipartite graphs stems from the fact that we would like to employ Proposition \ref{proposition:bipartite} later. The formal definition of \abjp{} is given below.

\defparproblem{\abjpfull{} (\abjp)}{A bipartite multi-graph $G$ with bipartition $(P,Q)$, $A,B \subseteq V(G)$ such that $A \cap B = \emptyset$, and integers $\mu,k_1$ and $k_2$.}{$k_1 + k_2$}{Does there exist a partition $(V_1,V_2)$ of $V(G)$ such that $A \subseteq V_1$, $B \subseteq V_2$, $|V_1| = \mu$ and for $i \in \{1,2\}$, $|E(G[V_i])| \leq k_i$?}

An instance of \abjp{} is denoted by $\abjp(G,A,B,\mu,k_1,k_2)$. A partition $(V_1,V_2)$ satisfying the above properties is called a \emph{witnessing coloring} of $\abjp(G,A,B,\mu,k_1,k_2)$. Furthermore, we need  the following theorem, proven later in this paper.

\begin{theorem}\label{thm:abjp}
\abjp{} can be solved in time $2^{k^{\OO(1)}}  \cdot n^{\OO(1)}$.
\end{theorem}

Let us now turn to focus on the proof of Theorem \ref{thm:jp}. 

\begin{proof}[Proof of Theorem \ref{thm:jp}]
Given an instance $\jp(G,\mu,k_1,k_2)$, call the algorithm given by Proposition \ref{prop:oct} with the instance  $\oct(G,k)$ as input.

\begin{claim}\label{claim:jp:1}
If $\oct(G,k)$ is a \no{}-instance, then $\jp(G,\mu,k_1,k_2)$ is a \no{}-instance.
\end{claim}
{\noindent\bf Proof.}
Suppose $\jp(G,\mu,k_1,k_2)$ is a \yes{}-instance. Let $(V_1,V_2)$ be a witnessing coloring for this instance. Let $E' = E(G[V_1]) \cup E(G[V_2])$. Then, observe that $G \setminus E'$ is a bipartite graph. Let $V'$ be  a set of vertices of  minimum size such that every edge in $E'$ has at least one endpoint in $V'$. Since $|E'| \leq k$, it holds that $|V'| \leq k$. Moreover, $G \setminus V'$ is bipartite. Therefore, $V'$ is an odd cycle transversal of $G$ of size at most $k$. Thus, $\oct(G,k)$ is a \yes{}-instance.
$\diamond$
\medskip

\begin{figure}[t]
\begin{center}
\includegraphics[scale=0.3]{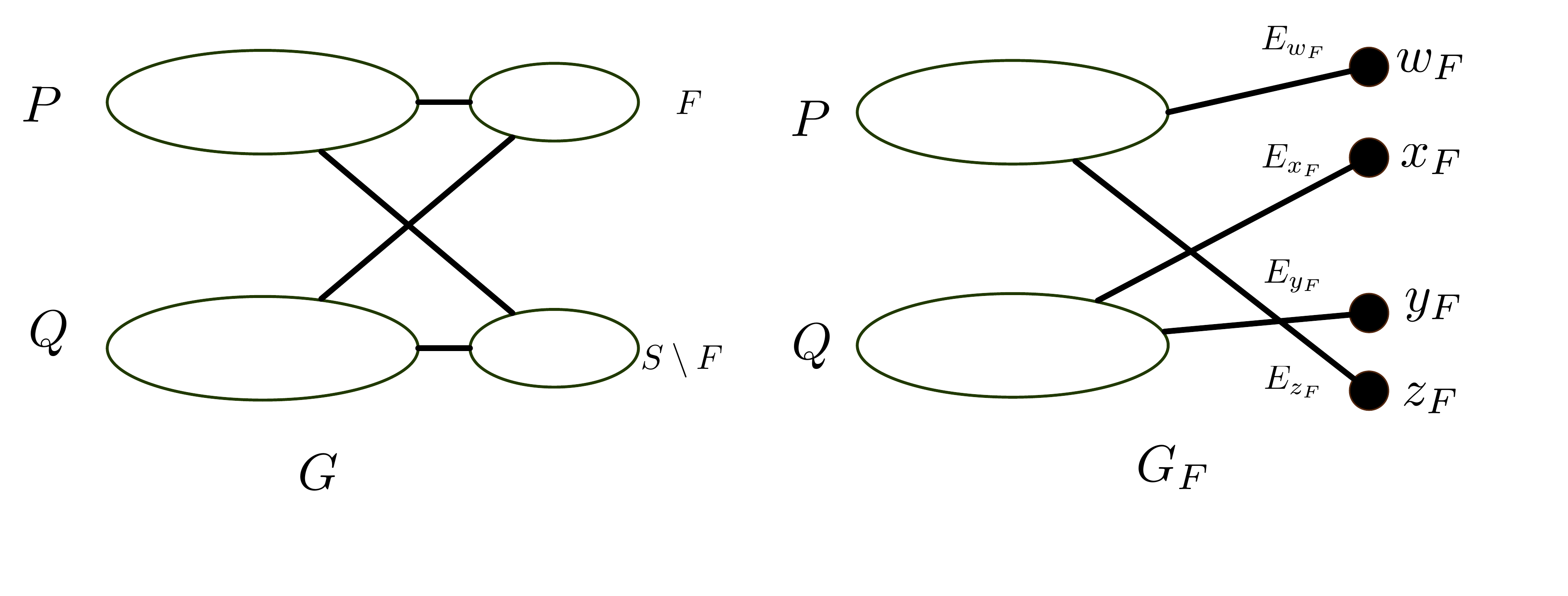}
\end{center}
\vspace{-1cm}
\caption{\label{fig:jp} The construction in the proof of Theorem \ref{thm:jp}.}
\end{figure}

Henceforth, let $S$ be an odd cycle transversal of $G$ of size at most $k$. Then, $G\setminus S$ is a bipartite graph. Fix some bipartition $(P,Q)$ of $G \setminus S$. Let $\mathcal{F}$ be the family of all subsets of $S$, that is, $\mathcal{F} = 2^S$. For any $F \in \mathcal{F}$, denote $l_1^F = |E(G[F])|$ and $l_2^F = |E(G[S \setminus F])|$, and let $G_F$ be the graph constructed as follows (see Fig.~\ref{fig:jp}).
\begin{itemize} 
\item $V(G_F) = V(G \setminus S) \cup \{w_F,x_F,y_F,z_F\}$, where $w_F,x_F,y_F,z_F$ are new distinct vertices.
\item $E(G_F) = E(G \setminus S) \cup E_{w_F} \cup E_{x_F} \cup E_{y_F} \cup E_{z_F}$, where the {\em multisets} $E_{w_F}, E_{x_F}, E_{y_F}$ and $E_{z_F}$ are defined as follows.
\begin{itemize}
\item $E_{w_F} = \{e_u=(w_F, u) :  u \in P, \text{ and there exists } v \in F \text{ such that } (u,v) \in E(G)\},$
\item $E_{x_F} = \{e_u(x_F, u) :  u \in Q, \text{ and there exists } v \in F \text{ such that } (u,v) \in E(G)\},$
\item $E_{y_F} = \{e_u=(y_F, u) :  u \in Q, \text{ and there exists } v \in S \setminus F \text{ such that } (u,v) \in E(G)\},$
\item $E_{z_F} = \{e_u=(z_F, u) :  u \in P, \text{ and there exists } v \in S \setminus F \text{ such that } (u,v) \in E(G)\}.$
\end{itemize}
\end{itemize}

Observe that $G_F$ is a bipartite graph with $(P \cup \{x_F,y_F\}, Q \cup \{w_F,z_F\})$ as a bipartition. 

\begin{claim}\label{claim:jp:equivalent}
$\jp(G,\mu,k_1,k_2)$ is a \yes{}-instance if and only if there exists $F \in \mathcal{F}$ such that $\abjp(G_F, \allowbreak  \{w_F,x_F\}, \allowbreak \{y_F,z_F\}, \mu - |F| + 2, k_1-l_1^F, k_2-l_2^F)$ is a \yes{}-instance.
\end{claim}
{\noindent\bf Proof.}
In the forward direction, suppose that $\jp(G,\mu, k_1,k_2)$ is a \yes{}-instance, and let $(V_1,V_2)$ be a witnessing coloring for $\jp(G,\mu,k_1,k_2)$. Moreover, let $F = V_1 \cap S $.
Now, we define a partition $(V'_1,V'_2)$ of $V(G_F)$ as follows: $V_1'=(V_1\setminus S)\cup\{x_F,y_F\}$ and $V_2'=(V_2\setminus S)\cup\{w_F,z_F\}$. Let us now argue that $(V_1',V_2')$ is a witnessing coloring for $\abjp(G_F, \{w_F,x_F\}, \{y_F,z_F\}, \mu - |F|+2,k_1-l_1^F, k_2-l_2^F)$. First, by the construction of $(V_1',V_2')$, we have that $\{x_F,y_F\} \subseteq V_1'$ and $\{w_F,z_F\} \subseteq V_2'$. Second, as $V_1'=(V_1\setminus S)\cup\{x_F,y_F\}$, we also have that $|V_1'| = |V_1| - |F| + 2 = \mu + |F|+2$. Third, observe that for any $i \in \{1,2\}$, $|E(G[V'_i])| = |E(G[V_i])| - |E(G[F])|$. Thus, $|E(G[V_i])| \leq k_i - l_{i}^{F}$. 

In the backward direction, suppose that there exists an $F \in \mathcal{F}$ such that $\abjp(G_F, \{w_F,$ $x_F\}, \{y_F,z_F\}, \allowbreak \mu - |F|+2,  k_1-l_1^F, k_2-l_2^F)$ is a \yes{}-instance, and let $(V'_1,V'_2)$ be a witnessing coloring for this instance. We now define a partition $(V_1,V_2)$ of $V(G)$ as follows:  $V_1=(V'_1\cap V(G))\cup F$ and $V_2=(V'_2\cap V(G))\cup (S\setminus F)$. Let us now argue that $(V_1,V_2)$ is a witnessing coloring for $\jp(G,\mu,k_1,k_2)$. From the definition of $V_1$, and since $V(G) = (V(G_F) \setminus \{w_F, x_F, y_F, z_F\}) \cup F$ and $F\cap V(G_F)=\emptyset$,  we have that $|V_1| = |V_1'| - |\{x_F, y_F\}| + |F| = \mu - |F| +2 -2 +|F| = \mu$.  Moreover, observe that $|E(G[V_1])| = |E(G[V_1'])| + |E(G[F])| \leq k_1 + l_1^F$ and $|E(G[V_2])| = |E(G[V_2'])| + |E(G[S \setminus F])| \leq k_2 + l_2^F$. This concludes the proof of the claim. $\diamond$

\medskip

Thus, to solve an instance of \jp, it is enough to solve $2^{|S|}\leq 2^k$ instances of \abjp. Hence, by Theorem \ref{thm:abjp}, \jp{} can be solved in time $2^{k^{\OO(1)}} n^{\OO(1)}$.
\end{proof}

\section{Solving \abjpfull}\label{sec:abjp}
Recall the problem definition of \abjpfull(\abjp){} from Section \ref{sec:jp}.
In this section, we prove Theorem \ref{thm:abjp}. For this purpose,  let us define another auxiliary problem, which we call \abcjpfull(\abcjp){}. Intuitively, \abcjp{} is exactly the same problem as \abjp{} where we are interested in an answer for every choice of $ \mu \in [n]_0$, $l_1 \in [k_1]_0$ and $l_2 \in [k_2]_0$, and additionally we demand the input graph to be connected.

\defparoutproblem{\abcjpfull (\abcjp)}{A connected bipartite multi-graph $G= (P,Q)$, $A,B \subseteq V(G)$ such that $A \cap B = \emptyset$, and integers $k_1$ and $k_2$.}{$k_1 + k_2$}{For all $\mu\in [n]_0$, $l_1 \in [k_1]_0$ and $l_2 \in [k_2]_0$, output a binary value, {\ttfamily aJP}$[\mu,l_1,l_2]$, which is $1$ if and only if there exists a partition $(V_1,V_2)$ of $V(G)$ such that
\begin{itemize}
 \item $A \subseteq V_1$ and $B \subseteq V_2$,
 \item $|V_1| = \mu$, and
 \item for $i \in \{1,2\}$, $|E(G[V_i])| \leq l_i$.
\end{itemize}
}

For any $\mu \in [n]_0$, $l_1 \in [k_1]_0$, $l_2 \in [k_2]_0$, a partition witnessing that  {\ttfamily aJP}$[\mu, l_1,l_2] =1$ is called a {\em witnessing coloring} for {\ttfamily aJP}$[\mu, l_1,l_2] =1$. Moreover, an instance of \abcjp{} is denoted by $\abcjp(G,A,B,k_1,k_2)$. In the rest of this paper, we prove the following~theorem.

\begin{theorem}\label{thm:abcjp}
\abcjp{} can be solved in time $2^{k^{\OO(1)}}  \cdot n^{\OO(1)}$.
\end{theorem}

Having Theorem \ref{thm:abcjp} at hand, a simple application of the method of dynamic programming results in the proof of Theorem \ref{thm:abjp}.

\noindent{\bf Proof of Theorem \ref{thm:abjp}.}\label{app:abjp}
Let $\abjp(G, A,B, \mu, k_1,k_2)$ be an instance of \abjp. 
%If $G$ is a connected graph, then $\abjp(G, A,B, k_1,k_2)$ is a \yes{} instance if and only if for the instance $\abcjp(G,A,B, k_1,k_2)$, {\ttfamily aJP}$[\mu, k_1,k_2] =1$.
%or $ANS-ABCJP[\lceil \frac{|V(G)|}{2} \rceil, k_1,k_2] =1$.  
%Otherwise, 
Let $C_1, \ldots, C_r$ be the connected components of $G$. For all $i \in [r]$, let $A_i = A \cap C_i$ and $B_i = B \cap C_i$. Let $I_i = \abcjp(C_i, A_i,B_i,k_1,k_2)$. Let {\ttfamily aJP}$_i$ be the output table for the instance $I_i$, returned by the algorithm of Theorem \ref{thm:abcjp}.
For any $j \in [r]$, let $G_j = G[\bigcup\limits_{i \in [j]} C_i]$. Note that $G = G_r$.
Let us define a 4-dimensional binary table {\ttfamily M} in the following way. For all $i \in [r]$, $\mu' \in [|V(G)|]_0$, $l_1 \in [k_1]_0$ and $l_2 \in [k_2]_0$, $\text{{\ttfamily M}}[i,\mu',l_1,l_2] = 1$ if and only if there exists a partition $(V_1,V_2)$ of $V(G_i)$ such that $(A \cap G_i) \subseteq V_1$, $(B \cap G_i) \subseteq V_2$, $|V_1| = \mu'$ and for $j \in \{1,2\}$, $|E(G[V_j])| \leq l_j$.
Observe that $\abjp(G,A,B,\mu, k_1,k_2)$ is a \yes{}-instance if and only if $\text{{\ttfamily M}}[r,\mu,k_1,k_2] =1$.
%or $ANS[r,\lceil \frac{|V(G)|}{2} \rceil,k_1,k_2] =1$.
%\begin{lemma}
%For all ${\mu}^{'} \in \{0, \ldots, \mu \}$, $l_1 \in \{0, \ldots, k_1\}$, $l_2 \in \{0, \ldots,k_2\}$, and for any $i \in \{1, \ldots, r\}$, one can compute $\text{{\ttfamily M}}[i,\mu,l_1,l_2]$ using the following recurrences.
We now compute $\text{{\ttfamily M}}[r,\mu,k_1,k_2]$ recursively using the following recurrences.

$$\text{{\ttfamily M}}[1,\mu',l_1,l_2]=\text{{\ttfamily aJP}}_1(\mu',l_1,l_2)$$
For all $i \in \{2, \ldots,r\}$, $\mu' \in [|V(G)|]_0$, $l_1 \in [k_1]_0$ and $l_2 \in [k_2]_0$,
$$\text{{\ttfamily M}}[i,\mu',l_1,l_2] = \bigvee_{\substack{{\mu}' = {\mu}^1 + {\mu}^2 \\ l_1 = l_{1}^1 + l_{1}^2 \\ l_2 = l_{2}^1 + l_{2}^2}} (\text{{\ttfamily M}}[i-1, {\mu}^1, l_{1}^1,l_{2}^1] \wedge \text{{\ttfamily aJP}}_i[{\mu}^2, l_{1}^2,l_{2}^2]),$$
where for all $j \in \{1,2\}$, ${\mu}^j$, $l_{1}^j$ and $l_{2}^j$ are non-negative integers.

Note that the time taken to compute $\text{{\ttfamily M}}[r,\mu,k_1,k_2]$ is at most $(r \cdot n^2 \cdot k_1^2 \cdot k_2^2 \cdot \tau)$, where $\tau$ is the time taken to solve an instance of \abcjp. Since from Theorem \ref{thm:abcjp}, an instance of \abcjp{} can be solved in time $2^{k^{\OO(1)}} \cdot n^{\OO(1)}$ and $r \leq n$, \abjp{} can be solved in time $2^{k^{\OO(1)}} \cdot n^{\OO(1)}$.  \qed
\section{Solving \abcjpfull}\label{sec:abcbjb}
Recall the problem definition of \abcjp{} from Section \ref{sec:abcbjb}. In this section, we prove Theorem \ref{thm:abcjp}. Let us start by stating a known result that is a crucial component of our proof. By this result, we would have an algorithm that efficiently computes a special type of tree decomposition, that we call a {\em highly connected tree decomposition}, where every bag is ``highly-connected'' rather than ``small'' as in the case of standard tree decompositions. While this property is the main feature of this decomposition, it is also equipped with other beneficial properties, such as a (non-trivial) upper bound on the size of its adhesions, which are {\em all} exploited by our algorithm.

\begin{theorem}[\cite{cygan2014minimum}]\label{theorem:std}
There exists an $2^{\OO(k^2)} n^2 m$-time algorithm that, given a {\em connected} graph $G$ together with an integer $k$, computes a tree decomposition $(T, \beta)$ of $G$ with at most $n$ nodes such that the following conditions hold, where $\eta = 2^{\OO(k)}$.

\begin{enumerate}
    \item\label{item:Unbreak1} For each $t \in V(T)$, the graph $G[\gamma(t) \setminus \sigma(t)]$ is connected and $N(\gamma(t) \setminus \sigma(t)) = \sigma(t)$.
    
    \item\label{item:Unbreak2} For each $t \in V(T)$, the set $\beta(t)$ is $(\eta, k)$-unbreakable in $G[\gamma(t)]$.
    
    \item\label{item:Unbreak3} For each non-root $t \in V(T)$, we have that $|\sigma(t)| \leq \eta$ and $\sigma(t)$ is $(2k,k)$-unbreakable in $G[\gamma(parent(t))]$.
    
\end{enumerate}
\end{theorem}

In order to process such a tree decomposition in a bottom-up fashion, relying on the method of dynamic programming, we need to address a specific problem associated with every bag, called \hpfull{} (\hp). We chose the name \hp{} to be consistent with the choice of problem name in \cite{cygan2014minimum}, yet we stress that our problem is more general than the one in \cite{cygan2014minimum} (since the handling of a bag in our case is more intricate than the one in \cite{cygan2014minimum}).

Roughly speaking, an input of \hp{} would consist of the following components. First, we are given ``budget'' parameters $k_1$ and $k_2$ as in an instance of \abcjp. Second, we are given an argument $b$ which would simply be $n$ (to upper bound  $|\gamma(t)|$) when we construct an instance of \hp\ while processing some node $t$ in the tree decomposition. Third, we are given a hypergraph $H$ which would essentially be the graph $G[\beta(t)]$ to which we add hyperedges that are supposed to represent the sets $\sigma(\widehat{t})$ for the children $\widehat{t}$ of $t$. Fourth, we are given an integer $q$ whose purpose is clarified in the discussion below the definition of \hp\ (in Definition \ref{def:global}). Finally, for every hyperedge $F$, we are given a function $f_F : [k]_0^F \times [b]_0 \times [k_1]_0 \times [k_2]_0 \to \{0,1\}$. To roughly understand the meaning of this function, first recall that $F$ is supposed to represent $\sigma(\widehat{t})$ for some child $\widehat{t}$ of $t$. Now, the function $f_F$ aims to capture all information obtained while we processed the child $\widehat{t}$ of $t$ that might be relevant to the node $t$. In particular, let us give an informal, intuitive interpretation of an element $(\Gamma, \mu, l_1,l_2)$ in the domain of $f_F$. For this purpose, note that when we remove at most $k$ edges from the (connected) graph $G[\gamma(\widehat{t})]$, we obtain at most $k+1$ connected components. The function $\Gamma$ can be thought of as a method to assign to each vertex in $\sigma(\widehat{t})$ the connected component in which it should lie. Such information is extremely useful since each such connected component is in particular a bipartite graph, and hence by relying on Proposition \ref{proposition:bipartite} and an exhaustive search, we would be able to use it to extract a witnessing coloring for an instance of \abcjp. The arguments $\mu, l_1$ and $l_2$ can be thought of as those in the definition of an output of $\abcjp$. Now, the value $f_F(\Gamma, \mu, l_1,l_2)$ aims to indicate whether $\Gamma, \mu, l_1$ and $l_2$ are ``realizable'' in the context of the child $\widehat{t}$ (the precise meaning of this value would become clearer later, once we establish additional necessary definitions.)

Let us now give the formal definition of \hp. In this definition, we denote $k=k_1+k_2$.

\defproblem{\hpfull{} (\hp)}{Integers $k_1$, $k_2$, $b$, $d$ and $q$, a multi-hypergraph $H$ with hyperedges of size at most $d$, and for all $F \in E(H)$, a function $f_F : [k]_0^F \times [b]_0 \times [k_1]_0 \times [k_2]_0 \to \{0,1\}$.}{For all $0 \leq \mu \leq b$, $0 \leq l_1 \leq k_1$, $0 \leq l_2 \leq k_2$, output the binary value

$$\text{{\ttfamily aHP}}[\mu,l_1,l_2] = \bigvee_{\Upsilon : V(H) \to [k]_0}\bigvee_{\substack{\{\mu^{F}\}|_{F \in E(H)} \\ \{l_{1}^{F}\}|_{F \in E(H)} \\ \{l_{2}^{F}\}|_{F \in E(H)}}} \bigwedge\limits_{F \in E(H)} f_F(\Upsilon|_{F}, {\mu}^{F}, l_{1}^{F}, l_{2}^{F}),$$

where $\mu = \sum\limits_{F \in E(H)} {\mu}^{F}$, $\sum\limits_{F \in E(H)} l_{1}^{F} \leq l_1$, $\sum\limits_{F \in E(H)} l_{2}^{F} \leq l_2$ and each of ${\mu}^{F}$, $l_{1}^{F}$ and  $l_{2}^{F}$ is a non-negative integer.}

For a particular choice of $\mu$, $l_1$ and $l_2$, a function $\Upsilon$ witnessing that $\text{{\ttfamily aHP}}[\mu,l_1,l_2] =1$ is called a witnessing coloring for $\text{{\ttfamily aHP}}[\mu,l_1,l_2]$. An instance of \hpfull{} is denoted by $\hp(k_1,k_2,b,d,q,H,\{f_F\}|_{F \in E(H)})$.

Although we are not able to tackle \hp{} efficiently at its full generality, we are still able to solve those instances that are constructed when we would like to ``handle'' a single bag in a highly connected tree decomposition. For the sake of clarity, let us now address the beneficial properties that these instances satisfy individually,  where each of them ultimately aims to ease our search for a witnessing coloring. The first property, called local unbreakability,  unconditionally restricts the way a function $\Gamma : F \to [k]_0$, to be thought of as a restriction of the witnessing coloring we seek, can color a hyperedge $F$ so that the value of $f_F$ is 1.\footnote{In this context, it may be insightful to recall Lemma \ref{lemma:numOfUnbreakableFunctions}.}

\begin{definition}
[Local Unbreakability] An instance $\hp(k_1,k_2,b,d,q,H,\{f_F\}|_{F \in E(H)})$ is {\em locally unbreakable} if every $F \in E(H)$ satisfies the following property: for any $\Gamma : F \to [k]_0$ that is not $(3{k}^2,k)$-unbreakable, $f_F(\Gamma, \mu, l_1,l_2) =0$ for all $0 \leq \mu \leq b$, $0 \leq l_1 \leq k_1$ and $0 \leq l_2 \leq k_2$. 
\end{definition}

The second property, called connectivity, implies that if we would like to use a function $\Gamma : F \to [k]_0$ to color a hyperedge (as a restriction of a witnessing coloring) with more than one color, then we would have to ``pay'' at least 1 from our budget $l_1+l_2$.

\begin{definition}[Connectivity] An instance $\hp(k_1,k_2,b,d,q,H,\{f_F\}|_{F \in E(H)})$ is {\em connected} if every $F \in E(H)$ satisfies the following property: for any $\Gamma : F \to [k]_0$ for which there exist distinct $i,j \in [k]_0$ such that $|\Gamma^{-1}(i)|,|\Gamma^{-1}(j)| > 0$, it holds that $f_F(\Gamma, \mu,l_1,l_2)=1$ only if $l_1 + l_2 \geq 1$. 
\end{definition}

%\noindent 
The third property, called global unbreakability, directly restricts our ``solution space'' by implying that we only need to determine whether there exists a $(q,k)$-unbreakable witnessing coloring.

\begin{definition}[Global Unbreakability]\label{def:global} An instance $\hp(k_1,k_2,b,d,q,H,\{f_F\}|_{F \in E(H)})$ is {\em globally unbreakable} if for all $0 \leq \mu \leq b$, $0 \leq l_1 \leq k_1$, $0 \leq l_2 \leq k_2$: if $\text{{\ttfamily aHP}}[\mu,l_1,l_2]=1$, then there exists a witnessing coloring $\Upsilon : V(H) \to [k]_0$ that is $(q,k)$-unbreakable.
\end{definition}

An instance $HP(k_1,k_2,b,d,q,H,\{f_F\}|_{F \in E(H)})$ is called a \emph{favorable instance of \hp} if it is locally unbreakable, connected and globally unbreakable. For such instances we have the following theorem.
\begin{theorem}\label{thm:hp}
\hp{} on favorable instances is solvable in time $2^{\OO(\min (k,q) \log (k+q))} d^{\OO(k^2)} m^{\OO(1)}$.
\end{theorem}

The proof of this theorem is very technical, involving non-trivial analysis of a very ``messy'' picture obtained by guessing part of a hypothetical witnessing coloring via the method of color coding. We defer the proof of Theorem~\ref{thm:hp} to Section~\ref{sec:hp}.
%Due to space constraints, the details are omitted here but can we found in the full version of the paper \cite{lokshtanov2017judiciousfull}.

From now onwards, to simplify the presentation of arguments ahead with respect to \abcjp, we would abuse notation and directly define a witnessing coloring as a function rather than a partition. More precisely, the term witnessing coloring for {\ttfamily aJP}$[\mu, l_1,l_2] =1$ would refer to a function $col: V(G)\rightarrow\{V_1,V_2\}$ such that $A\subseteq V_1$, $B\subseteq V_2$, |$V_1|=\mu$ and for $i\in\{1,2\}$, $|E(G[V_i])| \leq l_i$. To proceed to our proof of Theorem \ref{thm:abcjp}, we first need to introduce an additional notation. Roughly speaking, this notation translates a coloring $\Upsilon$ of the form that witnesses some $\text{{\ttfamily aHP}}[\mu,l_1,l_2] =1$  to a coloring of the form that witnesses {\ttfamily aJP}$[\mu, l_1,l_2] =1$ via some tuple $\mathbf{v} \in {\{0,1\}}^{k+1}$. Formally,

\begin{definition}\label{def:translation}
For a tuple $\mathbf{v} \in {\{0,1\}}^{k+1}$,  bipartite graph $G$ with bipartition $(P,Q)$,  $A \subseteq V(G)$ and $\Upsilon : A \to [k]_0$,  define {\em $\widehat{\Upsilon}_{\mathbf{v}} : A \to \{V_1,V_2\}$} as follows.
\begin{itemize}
\item For all $v \in P$, $\widehat{\Upsilon}_{\mathbf{v}}(v) = V_1$ if and only if $\mathbf{v}[\Upsilon(v)]=0$.
\item For all $v \in Q$,  $\widehat{\Upsilon}_{\mathbf{v}}(v) = V_1$ if and only if $\mathbf{v}[\Upsilon(v)]=1$.
\end{itemize}
\end{definition}

Suppose we are given an instance $\abcjp(G, A,B,k_1,k_2)$. Fix some bipartition $(P,Q)$ of $G$. Let $(T, \beta)$ be the highly connected tree decomposition computed by the algorithm of Theorem \ref{theorem:std}, and let $r$ be the root of $T$. In what follows, $\eta=2^{\OO(k)}$ as in Theorem \ref{theorem:std}, and $q = (\eta + k)k$. We now proceed to define a binary variable that is supposed to represent the answer we would like to compute when we process the bag of a specific node of the tree. Hence, one of the arguments is a node $t$, and three additional arguments are $\mu \in [n]_0, l_1 \in [k_1]_0$ and  $l_2 \in [k_2]_0$. However, we cannot be satisfied with one answer, but need an answer for every possible ``interaction'' between the bag of $t$ and the bag of its parent $t'$. Thus, the definition also includes a coloring of $\sigma(t)$. The tuple $\mathbf{v} \in \{0,1\}^{k+1}$ is necessary for the translation process described in Definition \ref{def:translation} (the way in which we shall obtain such a ``right'' tuple later in the proof would essentially rely on brute-force).

\begin{definition}\label{def:meaningY}
Given $t \in V(T)$, a $(3{k}^2,k)$-unbreakable function ${\Upsilon}^{\sigma} : \sigma(t) \to [k]_0$, a tuple $\mathbf{v} \in \{0,1\}^{k+1}$, and integers $\mu \in [n]_0, l_1 \in [k_1]_0$ and  $l_2 \in [k_2]_0$, the binary variable $y[t,{\Upsilon}^{\sigma},\mathbf{v},\mu,l_1,l_2]$ is 1 if and only if there exists $\Upsilon : \gamma(t) \to [k]_0$ extending ${\Upsilon}^{\sigma}$ such that
\begin{enumerate}
    \item\label{item:Y1} The translation $\widehat{\Upsilon}_{\mathbf{v}}$ maps to $V_1$ exactly $\mu$ vertices, that is, $|\widehat{\Upsilon}_{\mathbf{v}}^{-1} (V_1)| = \mu$.

    \item\label{item:Y2} The translation $\widehat{\Upsilon}_{\mathbf{v}}$ maps $A \cap \gamma(t)$ to $V_1$ and $B \cap \gamma(t)$ to $V_2$, that is, $A \cap \gamma(t) \subseteq \widehat{\Upsilon}_{\mathbf{v}}^{-1} (V_1)$ and $B \cap \gamma(t) \subseteq \widehat{\Upsilon}_{\mathbf{v}}^{-1} (V_2)$.
    
    \item\label{item:Y3} For all $i \in \{1,2\}$, it holds that $|E(G[\widehat{\Upsilon}_{\mathbf{v}}^{-1} (V_i)])| \leq l_i$.
    
    \item\label{item:Y4} The set of edges between vertices receiving different colors by $\Upsilon$ is exactly the set of edges between vertices that are mapped to the same side by the translation $\widehat{\Upsilon}_{\mathbf{v}}$, that is, \[\bigcup\limits_{\substack{i,j \in [k]_0, i\neq j}} E({\Upsilon}^{-1}(i),{\Upsilon}^{-1}(j)) = E(G[\widehat{\Upsilon}_{\mathbf{v}}^{-1}(V_1)]) \cup E(G[\widehat{\Upsilon}_{\mathbf{v}}^{-1}(V_2)]).\]
\end{enumerate}  
\end{definition}

A function $\Upsilon$ as above is called a {\em witnessing coloring} for $y[t,{\Upsilon}^{\sigma},\mathbf{v},\mu,l_1,l_2]$. 

\begin{lemma}\label{lemma:1}
For any $\mu \in [n]_0$, $l_1 \in [k_1]_0$ and $l_2 \in [k_2]_0$, {\ttfamily aJP}$[\mu,l_1,l_2]=1$ if and only if there exists $\mathbf{v} \in \{0,1\}^{k+1}$ such that $y[r,\phi,\mathbf{v},\mu,l_1,l_2]=1$.
\end{lemma}
\begin{proof}
Let us prove the backward direction first. Let $\mathbf{v} \in \{0,1\}^{k+1}$ be such that $y[r,\emptyset,\mathbf{v},\mu,l_1,l_2]=1$ and let $\Upsilon: V(G) \to [k]_0$ be one of its witnessing coloring. Then, Definition \ref{def:meaningY} directly implies that $\widehat{\Upsilon}_{\mathbf{v}}$ is a witnessing coloring for {\ttfamily aJP}$[\mu,l_1,l_2]=1$.

For the forward direction, let $col : V(G) \to \{V_1,V_2\}$ be a witnessing coloring for {\ttfamily aJP}$[\mu,l_1,l_2]$. Let $X= E(G[{col}^{-1}(V_1)]) \cup E(G[{col}^{-1}(V_2)])$. Let $C_0, \ldots, C_r$ be the connected components of $G \setminus X$. Since $X \subseteq E(G)$ and $|X| \leq l_1 + l_2 \leq k_1+k_2 = k$, we have that the number of connected components $r$ is upper bounded by $k$. For any $i \in [r]_0$, let $(P_i= (P \cap C_i),Q_i=(Q \cap C_i))$ be a bipartition of $C_i$ (recall that $G$ is a connected bipartite graph with fixed bipartition $(P,Q)$).

\begin{claim}\label{claim:1}
For any $i \in [r]_0$, either both $P_i \subseteq {col}^{-1}(V_1)$ and $Q_i \subseteq {col}^{-1}(V_2)$ or both $P_i \subseteq {col}^{-1}(V_2)$ and $Q_i \subseteq {col}^{-1}(V_1)$.
\end{claim}

{\noindent\bf Proof.}
Consider a bipartition $({P'}_i,{Q'}_i)$ of $C_i$, where ${P'}_i = {col}^{-1}(V_1)$ and ${Q'}_i = {col}^{-1}(V_2)$. Since $C_i$ is connected, from Proposition \ref{proposition:bipartite}, either $P_i \subseteq {P'}_i$ and $Q_i \subseteq {Q'}_i$, or $P_i \subseteq {Q'}_i$ and $Q_i \subseteq {P'}_i$. Hence the claim follows. $\diamond$
\medskip

Let us now construct a $k$-length binary string, $\mathbf{v}$, as follows. For any $i \in [r]_0$, $\mathbf{v}[i]=0$ if and only if $P_i \subseteq {col}^{-1}(V_1)$ and $Q_i \subseteq {col}^{-1}(V_2)$. For $i \in \{r+1, \ldots,k\}$, $\mathbf{v}[i]=0$.

Define $\Upsilon: V(G) \to [k]_0$ as follows. For any $v \in V(G)$, $\Upsilon(v) = i$ if and only if $v\in C_i$.

\begin{claim}\label{claim:2}
$\widehat{\Upsilon}_{\mathbf{v}} = col$.
\end{claim}
{\noindent\bf Proof.}
Consider some vertex $v \in V(G)$. Denote $V_j=col(v)$, $i=\Upsilon(v)$ and $b=\mathbf{v}[i]$, and note that $j \in \{1,2\}$, $i\in [k]_0$ and $b\in \{0,1\}$. We divide the argument into two cases corresponding to whether $v \in P_i$ or $v \in Q_i$. 
Since $v \in {col}^{-1}(V_j)$, if $v \in P_i$, then by Claim \ref{claim:1}, $P_i \subseteq {col}^{-1}(V_j)$ and $Q_i \subseteq {col}^{-1}(V_{3-j})$. Thus, by the construction of $\mathbf{v}$, $b=j-1$. Hence, by the definition of $\widehat{\Upsilon}_{\mathbf{v}}$, $\widehat{\Upsilon}_{\mathbf{v}}(v)=V_j$.
Similarly, if $v \in Q_i$, then by Claim \ref{claim:1}, $Q_i \subseteq {col}^{-1}(V_j)$ and $P_i \subseteq {col}^{-1}(V_{3-j})$. Thus, by the construction of $\mathbf{v}$, $b=2-j$. Hence, by the definition of $\widehat{\Upsilon}_{\mathbf{v}}$, $\widehat{\Upsilon}_{\mathbf{v}}(v)=V_j$.

Since the choice of $v$ was arbitrary, by the definition of $\widehat{\Upsilon}_{\mathbf{v}}$, we have that $\widehat{\Upsilon}_{\mathbf{v}}(v)=V_j$. $\diamond$
\medskip

\begin{claim}\label{claim:3}
For the binary string $\mathbf{v}$ constructed as above, the function $\Upsilon$ constructed above is a witnessing coloring for $y[r, \emptyset, \mathbf{v}, \mu,l_1,l_2]=1$.
\end{claim}

{\noindent\bf Proof.}
Since $\widehat{\Upsilon}_{\mathbf{v}} = col$, from the definition of $col$, we have that $|\widehat{\Upsilon}_{\mathbf{v}}^{-1}(V_1)| = \mu$, $A \subseteq \widehat{\Upsilon}_{\mathbf{v}}^{-1} (V_1)$, $B \subseteq \widehat{\Upsilon}_{\mathbf{v}}^{-1} (V_2)$, and for all $i \in \{1,2\}$, $|E(G[\widehat{\Upsilon}_{\mathbf{v}}^{-1} (V_i)])| \leq l_i$. Observe that $\bigcup\limits_{i,j \in [k]_0, i \neq j} E({\Upsilon}^{-1}(i),{\Upsilon}^{-1}(j))=X$. Therefore, $\bigcup\limits_{i,j \in [k]_0, i \neq j} E({\Upsilon}^{-1}(i),{\Upsilon}^{-1}(j)) = E(G[\widehat{\Upsilon}_{\mathbf{v}}^{-1}(V_1)]) \cup E(G[\widehat{\Upsilon}_{\mathbf{v}}^{-1}(V_2)])$. Thus, $\Upsilon$ is a witnessing coloring for $y[r, \emptyset, \mathbf{v}, \mu,l_1,l_2]=1$. $\diamond$
\medskip

This concludes the proof of the lemma.
\end{proof}

By Lemma \ref{lemma:1}, it is sufficient to compute $y[r, \phi, \mathbf{v}, \mu,l_1,l_2]$ for all $\mu\in[n]$, \lfin\ and \lsin. 
To this end, we need to compute $y[t,{\Upsilon}^{\sigma},\mathbf{v},\mu,l_1,l_2]$ for every node $t \in V(T)$, function ${\Upsilon}^{\sigma} : \sigma(t) \to [k]_0$ that is $(3{k}^2,k)$-unbreakable, tuple $\mathbf{v} \in \{0,1\}^{k+1},$ and integers $\mu \in [n]_0, l_1 \in [k_1]_0$ and $l_2 \in [k_2]_0$. Here, we employ bottom-up dynamic programming over the tree decomposition $(T, \beta)$. Let us now zoom into the computation of $y[t,{\Upsilon}^{\sigma},\mathbf{v},\mu,l_1,l_2]$ for all $\mu\in[n]$, \lfin\ and \lsin{}, for some specific $t,{\Upsilon}^{\sigma}$ and $\mathbf{v}$. Note that we now assume that values corresponding to the children of $t$ (if such children exist) have been already computed correctly. Moreover, note that $|\sigma(t)| \leq \eta$, the number of $(3{k}^2,k)$-unbreakable functions ${\Upsilon}^{\sigma} : \sigma(t) \to [k]_0$ is at most $|\eta|^{k^{\OO(1)}}= 2^{k^{\OO(1)}}$ (by Lemma \ref{lemma:numOfUnbreakableFunctions}), and the number of binary vectors of size $k+1$ is at most $2^{k+1}$. Thus, the total running time would consist of the computation time of $(T,\beta)$, and $n \cdot q^{\OO(k)} \cdot 2^{k+1}$ times the computation time for a set of values as the one we examine now. Hence, it remains to show how to compute the current set of values in time $2^{k^{\OO(1)}}$.

To compute our current set of values, let us construct an instance $\hp(k_1,k_2,n, \eta, q, H,$ $\{f_F\}|_{F \in E(H)})$ of \hp\ where $V(H) = \beta(t)$, and $E(H)$ and $\{f_F\}|_{F \in E(H)}$ are defined as follows.

\begin{enumerate}
    \item {\bf Type-1 Hyperedges.} For all $v \in \beta(t)$, insert $F=\{v\}$ into $E(H)$. Define $f_F : [k]^F_0 \times [n]_0 \times [k_1]_0 \times [k_2]_0 \to \{0,1\}$ as
    \[
    f_F(\Gamma, \mu,l_1,l_2) =
    \begin{cases}
     0,              & \text{if } v \in \sigma(t) \text{ and } \Gamma(v) \neq {\Upsilon}^{\sigma}(v) \\
     1,              & \text{if } v \in A,\ \widehat{\Gamma}_{\mathbf{v}}(F)=V_1,\ l_1=l_2=0 \text{ and } \mu=1 \\   
     1,              & \text{if } v \in B,\ \widehat{\Gamma}_{\mathbf{v}}(F)=V_2,\ l_1=l_2=0 \text{ and } \mu=0 \\ 
     1,              & \text{if } v \not \in A\cup B,\ l_1=l_2=0 \text{ and } \mu=[\widehat{\Gamma}_{\mathbf{v}}(F)=V_1] \\   
     0,              & \text{otherwise}
    \end{cases}
    \]

    \item {\bf Type-2 Hyperedges.} For all $(u,v) \in E(G[\beta(t)])$, add $F=\{u,v\}$ in $E(H)$. Define $f_F : [k]^F_0 \times [n]_0 \times [k_1]_0 \times [k_2]_0 \to \{0,1\}$ as
    \[
    f_F(\Gamma, \mu,l_1,l_2) =
    \begin{cases}
     0,              & \text{if } \mu \neq 0 \\
     1,              & \text{if } \widehat{\Gamma}_{\mathbf{v}}(u) \neq \widehat{\Gamma}_{\mathbf{v}}(v) \text{ and } \Gamma(u) = \Gamma(v) \\
     1,              & \text{if } \widehat{\Gamma}_{\mathbf{v}}(u) = \widehat{\Gamma}_{\mathbf{v}}(v) =V_1 \text{ and } l_1 \geq 1 \\
     1,              & \text{if } \widehat{\Gamma}_{\mathbf{v}}(u) = \widehat{\Gamma}_{\mathbf{v}}(v) =V_2 \text{ and } l_2 \geq 1 \\
     0,              & \text{otherwise}
    \end{cases}
    \]
    
    \item {\bf Type-3 Hyperedges.} For all $\widehat{t} \in V(T)$ that is a child of $t$ in the tree $T$, insert $F=\sigma(\widehat{t})$ into $E(H)$. Define $f_F : [k]^F_0 \times [n]_0 \times [k_1]_0 \times [k_2]_0 \to \{0,1\}$ as
    \[
    f_F(\Gamma, \mu,l_1,l_2) =
    \begin{cases}
     0,              & \text{if } \Gamma \text{ is not } (3{k}^2,k)\text{-unbreakable or } y[\widehat{t},\Gamma, \mu + {\mu}', l_1 + l_1', l_2 + l_2'] =0 \\     
     1,              & \text{otherwise}
    \end{cases}
    \]
    where ${\mu}' =|\widehat{\Gamma}_{\mathbf{v}}^{-1}(V_1)|$, and $l_i'=|\{\{u,v\} \in E(G[\sigma(\widehat{t})]) : \widehat{\Gamma}_{\mathbf{v}}(u) = \widehat{\Gamma}_{\mathbf{v}}(v)=V_i\}|$ for $i\in[2]$.
\end{enumerate}

Let us first claim that witnessing colorings related to $\hp(k_1,k_2,n, \eta, q, H,$ $\{f_F\}|_{F \in E(H)})$ are useful in the sense that they can be extended to witnessing colorings for the binary values in which we are interested.

\begin{lemma}\label{claim:semiequivalent}
For all $\mu\in[n]$, \lfin, \and \lsin, if  $\text{{\ttfamily aHP}}[\mu,l_1,$ $l_2] = 1$, then  $y[t, {\Upsilon}^{\sigma}, \mathbf{v},\mu,l_1,l_2]=1$. In fact, for any witness $\Upsilon : \beta(t) \to [k]_0$ of $\text{{\ttfamily aHP}}[\mu,l_1,l_2]=1$, there exists a function ${\Upsilon}' : \gamma(t) \to [k]_0$ that extends ${\Upsilon}$ and witnesses $y[t,{\Upsilon}^{\sigma},\mathbf{v},\mu,l_1,l_2]=1$.
\end{lemma}
\begin{proof}
If $\text{{\ttfamily aHP}}[\mu,l_1,l_2] = 1$, let $\Upsilon : \beta(t) \to  [k]_0$ be a witnessing coloring for $\text{{\ttfamily aHP}}[\mu,l_1,$ $l_2] = 1$. Then, there exist $\mu = \sum\limits_{F \in E(H)} {\mu}^F$, $\sum\limits_{F \in E(H)} l_1^F \leq l_1$ and $\sum\limits_{F \in E(H)} l_2^F \leq l_2 $, such that for all $F \in E(H)$, $f_F({\Upsilon}|_{F}, {\mu}^F, l_1^F, l_2^F)=1$. In particular, the following conditions hold.

\begin{enumerate}
    \item\label{item:1proof16} Since for any \emph{type-1} hyperedge $F$, it holds that $f_F({\Upsilon}|_{F}, {\mu}^F, l_1^F, l_2^F)=1$, we overall have that $\Upsilon^{\sigma} \subseteq \Upsilon$, $A \cap \beta(t) \subseteq \widehat{\Upsilon}_{\mathbf{v}}^{-1}(V_1)$, $B \cap \beta(t) \subseteq \widehat{\Upsilon}_{\mathbf{v}}^{-1}(V_2)$ and $$\sum\limits_{F \text{ is a type-1 hyperedge }} {\mu}^F = |\widehat{\Upsilon}_{\mathbf{v}}^{-1}(V_1) \cap \beta(t)|.$$
    
 \item\label{item:2proof16} Since for any \emph{type-2} hyperedge $F$  and $i\in\{1,2\}$, it holds that $f_F({\Upsilon}|_{F}, {\mu}^F, l_1^F, l_2^F)=1$, we overall have that $$|E(G[\widehat{\Upsilon}_{\mathbf{v}}^{-1}(V_i) \cap \beta(t)])| \leq  \sum\limits_{F \text{ is a type-2 hyperedge}} l_i^F.$$
    
    \item For any \emph{type-3} hyperedge $F=\sigma(t_i)$, since $f_F({\Upsilon}|_{F}, {\mu}^F, l_1^F, l_2^F)=1$, we have that ${\Upsilon}|_{F}$ is $(3{k}^2,k)$-unbreakable and $y[t_i,{\Upsilon}|_{F}, {\mu}^{F} + \mu', l_1^F + l_1', l_2^F + l_2'] =1$, where $\mu' =|\widehat{\Upsilon}_{\mathbf{v}}^{-1}(V_1) \cap F|$, 
$l_1'=|\{(u,v) \in E(G[\sigma(t_i)])| \widehat{\Upsilon}_{\mathbf{v}}(u) = \widehat{\Upsilon}_{\mathbf{v}}(v)=V_1\}|$ and $l_2'=|\{(u,v) \in E(G[\sigma(t_i)]) | \widehat{\Upsilon}_{\mathbf{v}}(u) = \widehat{\Upsilon}_{\mathbf{v}}(v)=V_2\}|$.

We thus derive that there exists a witnessing coloring ${\Upsilon}^i : \gamma(t_i) \to [k]_0$ for the condition $y[t_i,{\Upsilon}|_{F}, {\mu}^{F} + \mu', l_1^F + l_1', l_2^F + l_2'] =1$. Specifically, the following conditions are satisfied.
\begin{enumerate}
\item\label{item:3proof16a} ${\Upsilon}^i$ extends ${\Upsilon}|_{F}$.
\item\label{item:3proof16b} $|\widehat{\Upsilon^i}_{\mathbf{v}}^{-1} (V_1)| = \mu^F + \mu'$.
\item\label{item:3proof16c} $A \cap \gamma(t_i) \subseteq \widehat{\Upsilon^i}_{\mathbf{v}}^{-1} (V_1)$ and $B \cap \gamma(t_i) \subseteq \widehat{\Upsilon^i}_{\mathbf{v}}^{-1} (V_2)$.
\item\label{item:3proof16d} $|E(G[\widehat{\Upsilon^i}_{\mathbf{v}}^{-1}(V_1) \cap \gamma(t_i)]) | \leq l_1^F + l_1'$, and $|E(G[\widehat{\Upsilon}_{\mathbf{v}}^{-1}(V_2) \cap \gamma(t_i)]) | \leq l_2^F + l_2'$.  
\item\label{item:3proof16e} $\displaystyle{\bigcup\limits_{\substack{\ell,j \in [k]_0, \ell\neq j}} E({\Upsilon^i}^{-1}(\ell),{\Upsilon}^{-1}(j)) = E(G[\widehat{\Upsilon^i}_{\mathbf{v}}^{-1}(V_1)]) \cup E(G[\widehat{\Upsilon^i}_{\mathbf{v}}^{-1}(V_2)])}$.
\end{enumerate}    
\end{enumerate}

\noindent Keeping the above items in mind, we proceed to identify a witnessing coloring for $y[t, {\Upsilon}^{\sigma}, \mathbf{v},\mu,$ $l_1,l_2]=1$. We construct such a coloring $\Upsilon' : \gamma(t) \to [k]_0$ as follows. For all $v\in\gamma(t)$, if $v\in\beta(t)$, then define $\Upsilon'(v)=\Upsilon(v)$, and otherwise there exists a unique child $t_i$ of $t$ such that $v\in\gamma(t_i)$, in which case we define $\Upsilon'(v)=\Upsilon^i(v)$. 
For the sake of clarity, let us extract the required argument to the proof of a separate claim.

\begin{claim}
The aforementioned $\Upsilon'$ is a witnessing coloring for $y[t, {\Upsilon}^{\sigma}, \mathbf{v},\mu,l_1,l_2]=1$.
\end{claim}

{\noindent\bf Proof.} First, note that by Item \ref{item:1proof16}, we have that ${\Upsilon}_{\sigma}\subseteq \Upsilon$ and therefore ${\Upsilon}_{\sigma}\subseteq \Upsilon'$. Let us now verify that all of the other conditions specified in Definition~\ref{def:meaningY} are satisfied.
\begin{itemize}    
    \item Let us first prove Condition \ref{item:Y1}. To this end, we observe that by Items \ref{item:1proof16}, \ref{item:3proof16a} and \ref{item:3proof16b}, we have that the three following equalities~hold.
			\begin{itemize}
			\item $|\widehat{\Upsilon'}_{\mathbf{v}}^{-1}(V_1)| = |\widehat{\Upsilon'}_{\mathbf{v}}^{-1}(V_1) \cap \beta(t)|+  \sum\limits_{t_i \text{ is a child of } t \text{ in } T}|\widehat{\Upsilon'}_{\mathbf{v}}^{-1}(V_1) \cap (\gamma(t_i) \setminus \sigma(t_i))|.$
			\item $|\widehat{\Upsilon'}_{\mathbf{v}}^{-1}(V_1) \cap \beta(t)| = |\widehat{\Upsilon}_{\mathbf{v}}^{-1}(V_1) \cap \beta(t)| = \sum\limits_{F \text{ is a type-1 hyperedge}} {\mu}^F$.
			\item For every child $t_i$ of $t$, $|\widehat{\Upsilon'}_{\mathbf{v}}^{-1}(V_1) \cap (\gamma(t_i) \setminus F)| = {\mu}^F$, where $F = \sigma(t_i)$.
			\end{itemize}
    Thus, since $\sum\limits_{F \text{ is a type-2 hyperedge}} {\mu}^F =0$, we conclude that $|\widehat{\Upsilon'}_{\mathbf{v}}^{-1}(V_1)| = \sum\limits_{F \in E(H)} {\mu}^F = \mu$.
    
    \item Next, we prove Condition \ref{item:Y2}. However, by Items \ref{item:1proof16} and \ref{item:3proof16c}, we directly deduce that both $A \cap \gamma(t) \subseteq \widehat{\Upsilon'}_{\mathbf{v}}^{-1}(V_1)$ and $B \cap \gamma(t) \subseteq \widehat{\Upsilon'}_{\mathbf{v}}^{-1}(V_2)$ as required.
    
\item We now turn to prove Condition \ref{item:Y3}. In light of Item \ref{item:3proof16a}, note that 
\begin{eqnarray*}
|E(G[\widehat{\Upsilon'}_{\mathbf{v}}^{-1}(V_1)])| &  = & 
|E(G[\widehat{\Upsilon'}_{\mathbf{v}}^{-1}(V_1) \cap \beta(t)])|+  \sum\limits_{t_i \text{ is a child of } t \text{ in } T}  |E(G[\widehat{\Upsilon'}_{\mathbf{v}}^{-1}(V_1) \cap \gamma(t_i)])|  \\
&& 
- \sum\limits_{t_i \text{ is a child of } t \text{ in } T}  |E(G[\widehat{\Upsilon'}_{\mathbf{v}}^{-1}(V_1) \cap \sigma(t_i)])|.
\end{eqnarray*}
    
 Now, observe that by Items \ref{item:2proof16}, \ref{item:3proof16a} and \ref{item:3proof16d}, the two following equations hold.
\begin{itemize}
\item $|E(G[\widehat{\Upsilon'}_{\mathbf{v}}^{-1}(V_1) \cap \beta(t)])| = 
|E(G[\widehat{\Upsilon}_{\mathbf{v}}^{-1}(V_1) \cap \beta(t)])| \leq  \sum\limits_{F \text{ is a type-2 hyperedge}} l_1^F$.
\item For every child $t_i$ of $t$, $|E(G[\widehat{\Upsilon'}_{\mathbf{v}}^{-1}(V_1) \cap \gamma(t_i)])| = l_1^F + |E(G[ \widehat{\Upsilon'}_{\mathbf{v}}^{-1}(V_1) \cap \sigma(t_i)])|$, where $F = \sigma(t_i)$.
\end{itemize} 
    Since $\sum\limits_{F \text{ is a type-1 hyperedge}} l_1^F =0$, we conclude that 
$$|E(G[\widehat{\Upsilon'}_{\mathbf{v}}^{-1}(V_1)])| \leq \sum\limits_{F \in E(H)} l_1^F \leq l_1 .$$
    
  Similarly, we derive that $|E(G[\widehat{\Upsilon'}_{\mathbf{v}}^{-1}(V_2)])| \leq \sum\limits_{F \in E(H)} |l_2^F| \leq l_2$.
    
    \item Finally, we prove Condition \ref{item:Y4}. In the first direction, consider some edge  $e \in E(G[\widehat{\Upsilon'}_{\mathbf{v}}^{-1}(V_1)])$ $\cup E(G[\widehat{\Upsilon'}_{\mathbf{v}}^{-1}(V_2)])$. Let us denote $e=\{u,v\}$, and observe that $\widehat{\Upsilon'}_{\mathbf{v}}(v)=\widehat{\Upsilon'}_{\mathbf{v}}(u)$. If $u,v \in \gamma(t_i)$ for some child $t_i$ of $t$, then by Item \ref{item:3proof16e}, we have that $e \in \bigcup\limits_{\substack{i,j \in [k]_0, \\ i \neq j}} E({\Upsilon'}^{-1}(i),{\Upsilon'}^{-1}(j))$. Otherwise, $u,v \in \beta(t)$, and thus  $e$ is some type-2 hyperedge $F$. Since $f_F({\Upsilon}|_{F}, {\mu}^F, l_1^F, l_2^F)=1$, the definition of $f_F({\Upsilon}|_{F}, {\mu}^F, l_1^F, l_2^F)$ directly implies that ${\Upsilon}(v) \neq {\Upsilon}(v)$, and therefore again $e \in \bigcup\limits_{\substack{i,j \in [k]_0, \\ i \neq j}} E({\Upsilon'}^{-1}(i),{\Upsilon'}^{-1}(j))$.
		
    In the other direction, consider some edge $e \in \bigcup\limits_{\substack{i,j \in [k]_0, \\ i \neq j}} E({\Upsilon'}^{-1}(i),{\Upsilon'}^{-1}(j))$. Let us denote $e=\{u,v\}$, and observe that ${\Upsilon'}(v)\neq {\Upsilon'}(u)$. If $u,v \in \gamma(t_i)$ for some child $t_i$ of $t$, then by Item \ref{item:3proof16e}, we have that $e \in E(G[\widehat{\Upsilon'}_{\mathbf{v}}^{-1}(V_1)]) \cup E(G[\widehat{\Upsilon'}_{\mathbf{v}}^{-1}(V_2)])$. Otherwise, $u,v \in \beta(t)$, and thus  $e$ is some type-2 hyperedge $F$. Since $f_F({\Upsilon}|_{F}, {\mu}^F, l_1^F, l_2^F)=1$, the definition of $f_F({\Upsilon}|_{F}, {\mu}^F, l_1^F, l_2^F)$ directly implies that $\widehat{\Upsilon'}_{\mathbf{v}}(v)=\widehat{\Upsilon'}_{\mathbf{v}}(u)$, and therefore again $e \in E(G\widehat{\Upsilon'}_{\mathbf{v}}^{-1}(V_1)]) \cup E(G[\widehat{\Upsilon'}_{\mathbf{v}}^{-1}(V_2)])$.
\end{itemize}

\noindent Thus, we have proved that $\Upsilon'$ is a witnessing coloring for $y[t, {\Upsilon}^{\sigma}, \mathbf{v},\mu,l_1,l_2]$. Moreover, $\Upsilon'$, which extends $\Upsilon$, is the desired function for the second part of the lemma. $\diamond$
\medskip

\noindent This concludes the proof of the lemma.
\end{proof}

In light of Lemma \ref{claim:semiequivalent}, we now turn to verify that $\hp(k_1,k_2,n, \eta, q, H,$ $\{f_F\}|_{F \in E(H)})$ is of the form that we are actually able to solve.

\begin{lemma}\label{claim:favourable}
$\hp(k_1,k_2,n, \eta, q, H,$ $\{f_F\}|_{F \in E(H)})$ is a favorable instance of \hp.
\end{lemma}
\begin{proof}
Let us verify that each of the three properties of a favorable instance is satisfied.

\begin{itemize}
    \item \textbf{Local Unbreakability:} Let us choose an arbitrary $F \in E(H)$. If $F$ is a type-1 or a type-2 hyperedge, then since $|F|\leq 2$, we have that local unbreakability is trivially satisfied. Otherwise, if $F$ is a type-3 hyperedge, then the satisfaction of local unbreakability directly follows from the construction of $f_F$.
		
    \item \textbf{Connectivity:} Choose an arbitrary $F \in E(H)$ along with a tuple $(\Gamma,\mu,l_1,l_2)$ in the domain of $f_F$ such that $f_F(\Gamma,\mu,l_1,l_2)=1$. If $F$ is a type-1 hyperedge, then connectivity trivially holds. If $F$ is a type-2 hyperedge, then connectivity follows from the construction of $f_F$. Indeed, to see this, let us denote $F=\{u,v\}$. Then, if $\Gamma(u)\neq\Gamma(v)$, by the second and last cases in the definition of $f_F$, we deduce that $\widehat{\Gamma}_{\mathbf{v}}(u) = \widehat{\Gamma}_{\mathbf{v}}(v)$, else we contradict the supposition that $f_F(\Gamma,\mu,l_1,l_2)=1$. Then, connectivity directly follows from the third and fourth cases.
		
Now, suppose that $F = \sigma(\widehat{t})$ is a type-3 hyperedge, and say $\Gamma : F \to [k]_0$ is such that there exist $i,j \in [k]_0$, $i \neq j$, satisfying $|{\Gamma}^{-1}(i)|>0$ and $|{\Gamma}^{-1}(j)| >0$. We need to show that $l_1 + l_2 \geq 1$. Since $f_F(\Gamma, \mu,l_1,l_2)=1$, it holds that $y[\widehat{t},\Gamma, \mu + \mu', l_1 + l_1', l_2 + l_2']=1$, where $\mu'$, $l_1'$ and $l_2'$ are as defined at the construction of $f_F$. Let $\Upsilon : \gamma(\widehat{t}) \to [k]_0$ denote some witnessing coloring for this condition. Since $(T,\beta)$ is a \emph{highly connected tree decomposition}, the Property~\ref{item:Unbreak1} of such a decomposition implies that $G^* = G[\gamma(\widehat{t})] \setminus E(G[\sigma(\widehat{t})])$ is connected and that every vertex in $\sigma(\widehat{t}$ has at least one vertex in $V(G^*)$ that is its neighbor in $G[\gamma(\widehat{t}]$. In particular, every two vertices in $\sigma(\widehat{t})$ are connected by a path in $G^{*}$ (observe that $V(G^*) = \gamma(\widehat{t})$ as we have only edges are discarded when $G[\gamma(\widehat{t})]$ is modified to be $G^*$). Let $u \in {\Gamma}^{-1}(i)$ and $v  \in {\Gamma}^{-1}(j)$. Note that $u \neq v$ and $i \neq j$. Since $u$ and $v$ are connected by a path in $G^{*}$, we derive that $G^*$ has an edge $e$ such that
$$e \in \left( \bigcup\limits_{c,d \in [k]_0, c \neq d} E({\Upsilon}^{-1}(c), {\Upsilon}^{-1}(d))\right)  \setminus E(G[\sigma(t')]).$$ 
Recall that 
$\displaystyle{\bigcup\limits_{c,d \in [k]_0, c \neq d} E({\Upsilon}^{-1}(c), {\Upsilon}^{-1}(d)) = E(G[\widehat{\Upsilon}_{\mathbf{v}}^{-1}(V_1)]) \cup E(G[\widehat{\Upsilon}_{\mathbf{v}}^{-1}(V_2)])}$. Therefore, we have that $e \in (E(G[\widehat{\Upsilon}_{\mathbf{v}}^{-1}(V_1)]) \cup E(G[\widehat{\Upsilon}_{\mathbf{v}}^{-1}(V_2)])) \setminus E(G[\sigma(\widehat{t})])$. Thus, by the inductive hypothesis, $l_1 + l_2 \geq 1$.
    
\item \textbf{Global Unbreakability:} Suppose that $\text{{\ttfamily aHP}}[\mu,l_1,l_2] =1$. Then, by Lemma \ref{claim:equivalent}, there exists $\Upsilon'  : \gamma(t) \to [k]_0$ satisfying the properties listed in that lemma. From here, we get that $\sum\limits_{\substack{i,j \in [k]_0, i < j}}|E({{\Upsilon'}}^{-1}(i) , {{\Upsilon'}}^{-1}(j))| \leq l_1 + l_2 \leq k_1 + k_2 \leq k$. We argue that ${\Upsilon'}|_{ \beta(t)}$ is a witnessing coloring for global unbreakability, that is, this function is $(q,k)$-unbreakable. In this context, we remind that $q=(\eta+k)k$. To prove our argument, we first prove the following claim.
    
   \begin{claim}\label{claim:favourable:gb}
   Suppose that there exists $i \in [k]_0$ such that $|{\Upsilon'}^{-1} (i) \cap \beta(t)| > \eta + k$. Then, $\sum\limits_{\substack{j \in [k]_0, i\neq j}} |{\Upsilon'}^{-1}(j) \cap \beta(t)| \leq \eta + k$.   
   \end{claim}
   
{\noindent\bf Proof.}
   Suppose that the claim is false. Then, both $|{\Upsilon'}^{-1} (i) \cap \beta(t)| > \eta + k$ and $\sum\limits_{j \in [k]_0, i\neq j} |{\Upsilon'}^{-1}(j)$ $\cap \beta(t)| > \eta + k$. Thus, 
$$\Bigg(X={\Upsilon'}^{-1} (i) \cap \beta(t), Y= \bigg(\bigcup\limits_{\substack{j \in [k]_0, i\neq j}} {\Upsilon'}^{-1}(j) \cap \beta(t)\bigg) \cup \delta({{\Upsilon'}^{-1} (i) \cap \beta(t)})\Bigg) $$ is a separation of order at most $k$ of $G[\gamma(t)]$ as we have already shown that $$\sum\limits_{\substack{i,j \in [k]_0, i \leq j}}|E({{\Upsilon'}}^{-1}(i) , {{\Upsilon'}}^{-1}(j))| \leq l_1 + l_2 \leq k_1 + k_2 \leq k.$$ Moreover, $|(X \setminus Y) \cap \beta(t)| > \eta$ and $|(Y \setminus X) \cap \beta(t)| > \eta$, which contradicts that $\beta(t)$ is $(\eta,k)$-unbreakable in $G[\gamma(t)]$. $\diamond$
   \medskip
	
Thus, if there exist $i\in[k]_0$ as defined in  Claim \ref{claim:favourable:gb}, then we are done. That is, we conclude that ${\Upsilon'}|_{\beta(t)}$ is a $(q,k)$-unbreakable. Otherwise, for all $i \in [k]_0$, it holds that $|{\Upsilon'}^{-1}(i)| \leq \eta + k$. In particular, for any $i \in [k]_0$, $\sum\limits_{\substack{j\in [k]_0, i \neq j}}|{\Upsilon'}^{-1}(j)| \leq (\eta + k)k = q$. Thus, we again conclude that ${\Upsilon'}_{|\beta(t)}$ is $(q,k)$-unbreakable.\qed
\end{itemize}
\end{proof}

Finally, we turn to address the statement complementary to the one of Lemma \ref{claim:semiequivalent}.

\begin{lemma}\label{claim:equivalent}
For all $\mu\in[n]$, \lfin\ and \lsin, if  $y[t, {\Upsilon}^{\sigma}, \mathbf{v},\mu,l_1,l_2]=1$, then $\text{{\ttfamily aHP}}[\mu,l_1,$ $l_2] = 1$.
\end{lemma}
\begin{proof}
Fix some $\mu\in[n]$, \lfin\ and \lsin\ such that $y[t, {\Upsilon}^{\sigma}, \mathbf{v},\mu,l_1,l_2]=1$. Our objective is to show that $\text{{\ttfamily aHP}}[\mu,l_1,l_2] = 1$. To this end, let $\Upsilon$ be a witnessing coloring for $y[t, {\Upsilon}^{\sigma}, \mathbf{v},\mu,l_1,l_2]=1$. We would like to prove that ${\Upsilon}|_{ \beta(t)}$ is a witnessing coloring for  $\text{{\ttfamily aHP}}[\mu,l_1,l_2] = 1$, which would complete the proof of the lemma. To do so, we proceed as follows.

First, for any hyperedge $F \in E(H)$, let us define ${\mu}^F$, ${l_1}^F$ and ${l_2}^F$ as follows. 

\begin{itemize}
    \item {\bf If $F$ is a type-1 hyperedge:} Set ${\mu}^F = 1$ if $\widehat{\Upsilon}_{\mathbf{v}}(F) = V_1$, and ${\mu}^F = 0$ otherwise. Set $l_1^F =0$ and $l_2^F =0$.
    
    \item {\bf If $F=\{u,v\}$ is a type-2 hyperedge:} Set ${\mu}^F = 0$. 
    If $\widehat{\Upsilon}_{\mathbf{v}}(u)  \neq \widehat{\Upsilon}_{\mathbf{v}}(v)$ and $\Upsilon(u)=\Upsilon(v)$, set ${l_1}^F ={l_2}^F=0$. Otherwise, if $\widehat{\Upsilon}_{\mathbf{v}}(u)  = \widehat{\Upsilon}_{\mathbf{v}}(v) = V_1$, set ${l_1}^F =1$ and ${l_2}^F=0$, and if $\widehat{\Upsilon}_{\mathbf{v}}(u) = \widehat{\Upsilon}_{\mathbf{v}}(v) = V_2$, set ${l_1}^F =0$ and ${l_2}^F=1$. The case where $\widehat{\Upsilon}_{\mathbf{v}}(u)  \neq \widehat{\Upsilon}_{\mathbf{v}}(v)$ and $\Upsilon(u)\neq \Upsilon(v)$ cannot arise. Indeed, since $\Upsilon$ is a witnessing coloring for $y[t, {\Upsilon}^{\sigma}, \mathbf{v},\mu,l_1,l_2]=1$, we have that $\displaystyle{\bigcup\limits_{\substack{i,j \in [k]_0, i \neq j}} E({\Upsilon}^{-1}(i),{\Upsilon}^{-1}(j)) = E(G[\widehat{\Upsilon}_{\mathbf{v}}^{-1}(V_1)]) \cup E(G[\widehat{\Upsilon}_{\mathbf{v}}^{-1}(V_2)])}$.
    
    \item {\bf If $F$ is a type-3 hyperedge:} Denote $F=\sigma(\widehat{t})$, where $\widehat{t}$ is a child of $t$ in T. Set ${\mu}^F =|\widehat{\Upsilon}_{\mathbf{v}}^{-1}(V_1) \cap (\gamma(\widehat{t}) \setminus \sigma(\widehat{t}))|$, $l_1^{F}=|E(G[\widehat{\Upsilon}_{\mathbf{v}}^{-1}(V_1) \cap \gamma(\widehat{t})])| -|E(G[\widehat{\Upsilon}_{\mathbf{v}}^{-1}(V_1) \cap \sigma(\widehat{t})])|$ and $l_2^{F}=|E(G[\widehat{\Upsilon}_{\mathbf{v}}^{-1}(V_2) \cap \gamma(\widehat{t})])| - |E(G[\widehat{\Upsilon}_{\mathbf{v}}^{-1}(V_2) \cap \sigma(\widehat{t})])|$.
\end{itemize}

\noindent Let us proceed by proving three claims that would together imply that ${\Upsilon}|_{ \beta(t)}$ is a witnessing coloring for  $\text{{\ttfamily aHP}}[\mu,l_1,l_2] = 1$.
\begin{claim}\label{claim:equivalent:1}
%\begin{claim}
Let $\widehat{t}$ be a child of $t$ in $T$, and let $i \in [k]_0$ be such that $|{\Upsilon}^{-1}(i) \cap \sigma (\widehat{t})| > 3k$. Then, $\sum\limits_{\substack{j \in [k]_0, i \neq j}}|{\Upsilon}^{-1}(j)\cap \sigma(\widehat{t})| \leq 3k $.
\end{claim}

{\noindent\bf Proof.}
%\begin{proof}
Suppose, by way of contradiction, that the claim is false.  That is, we have that both $|{\Upsilon}^{-1}(i) \cap \sigma(\widehat{t})| > 3k$ and $\sum\limits_{\substack{j \in [k]_0, i \neq j}} |{\Upsilon}^{-1}(j)\cap \sigma(\widehat{t})| > 3k$. Consider the separation $(X,Y)$ of $G[\gamma(t)]$, where $X={\Upsilon}^{-1}(i)$ and $Y= (\gamma(t) \setminus {\Upsilon}^{-1}(i)) \cup \delta({\Upsilon}^{-1}(i))$. Observe that $X \cap Y = \delta({\Upsilon}^{-1}(i))$. Since $\Upsilon$ is a witnessing coloring for $y[t, {\Upsilon}^{\sigma}, \mathbf{v},\mu,l_1,l_2]$, we have that
$$\bigcup\limits_{\substack{i,j \in [k]_0, i \neq j}} E({\Upsilon}^{-1}(i),{\Upsilon}^{-1}(j)) = E(G[\widehat{\Upsilon}_{\mathbf{v}}^{-1}(V_1)]) \cup E(G[\widehat{\Upsilon}_{\mathbf{v}}^{-1}(V_2)])$$ and $|E(G[\widehat{\Upsilon}_{\mathbf{v}}^{-1}(V_1)]) \cup E(G[\widehat{\Upsilon}_{\mathbf{v}}^{-1}(V_2)])| \leq l_1+l_2 \leq k_1+k_2 \leq k $. Therefore, $|\delta({\Upsilon}^{-1}(i))| \leq k$, and thus the order of separation $(X,Y)$ is at most $k$. Moreover, since $|{\Upsilon}^{-1}(i) \cap \sigma(\widehat{t})| > 3k$, we have that $|(X \setminus Y) \cap \sigma(\widehat{t})| > 3k -  k =2k$, and since $\sum\limits_{\substack{j \in [k]_0, i \neq j}} |{\Upsilon}^{-1}(j)\cap \sigma(\widehat{t})| > 3k$, we also have that $|(Y \setminus X) \cap \sigma(\widehat{t})| > 3k$. This implies that $\sigma(\widehat{t})$ is not $(2k,k)$-unbreakable in $G[\gamma(t)]$, which mean that $\sigma(\widehat{t})$ is not $(2k,k)$-unbreakable in $G[\gamma(\mathrm{parent}(\widehat{t})]$. This is a contradiction to the fact that $(T, \beta)$ is a highly connected tree decomposition---specifically, it should satisfy Property \ref{item:Unbreak3} in Theorem \ref{theorem:std}. $\diamond$
%\end{proof} 

\medskip
\noindent Having Claim \ref{claim:equivalent:1} at hand, we now verify that each function $f_F$ assigns 1 to the required~tuple.

\begin{claim}\label{claim:equivalent:2}
For any $F \in E(H)$, $f_F({\Upsilon}|_{F}, {\mu}^F, l_1^F, l_2^F) =1$.
\end{claim}
{\noindent\bf Proof.}
First, note that since $\Upsilon$ be a witnessing coloring for $y[t, {\Upsilon}^{\sigma}, \mathbf{v},\mu,l_1,l_2]=1$, we have that $\Upsilon \subseteq {\Upsilon}_{\sigma}$, $A \cap \gamma(t) \subseteq \widehat{\Upsilon}_{\mathbf{v}}^{-1}(V_1)$ and $B \cap \gamma(t) \subseteq \widehat{\Upsilon}_{\mathbf{v}}^{-1}(V_2)$. Thus, from the construction of a type-1 hyperedge $F$ and the corresponding function $f_F$ with respect to $HP(k_1,k_2,n,\eta,q,H,{(f_F)}_{F \in E(H)})$, it is clear that $f_F({\Upsilon}|_{F},{\mu}^F,l_1^F,l_2^F)=1$.
Second, suppose $F$ is a type-2 hyperedge. The specifications of $f_F$, together with our  definition of ${\mu}^F,l_1^F$ and $l_2^F$, directly implies that $f_F({\Upsilon}|_{F},{\mu}^F,l_1^F,l_2^F)=1$.
%From the construction of \emph{type-1} and \emph{type-2} hyperedge $F$ and their corresponding $f_F$ for the instance $HP(k_1,k_2,n,\eta,q,H,{(f_F)}_{F \in E(H)})$, it is clear that $f_F({\mu}^F,{l_1}^F,{l_2}^F)=1$. 

Third, suppose that $F$ is a type-3 hyperedge, and denote $F=\sigma(t_i)$ for some $t_i$ that is a child of $t$ in $T$. Note that $y[t_i, {\Upsilon}|_{F}, \mathbf{v}, {\mu}^{F} + \mu', l_1^F + l_1', l_2^{F} + l_2']=1$ because ${\Upsilon}|_{\gamma(t_i)}$ is a witnessing coloring for this equality, where $\mu' = |\widehat{\Upsilon}_{\mathbf{v}}^{-1}(V_1) \cap \sigma(\widehat{t})|$, $l_1'=|E(G[\widehat{\Upsilon}_{\mathbf{v}}^{-1}(V_1) \cap \sigma(\widehat{t})])|$ and $l_2'=|E(G[\widehat{\Upsilon}_{\mathbf{v}}^{-1}(V_2) \cap \sigma(\widehat{t})])|$.
We now need to show that ${\Upsilon}|_{F}$ is $(3{k}^2,k)$-unbreakable, as then we would be able to conclude that $f_F({\Upsilon}|_{F}, {\mu}^F, l_1^F, l_2^F) =1$. By Claim \ref{claim:equivalent:1}, if there exists $i \in [k]_0$ such that $|{\Upsilon}^{-1}(i) \cap \sigma (\widehat{t})| > 3k$, then we deduce that ${\Upsilon}|_{\sigma(\widehat{t})}$ is $(3{k}^2,k)$-unbreakable. Otherwise, for all 
$i \in [k]_0$, $|{\Upsilon}^{-1}(i) \cap \sigma (\widehat{t})| \leq 3k$. Hence, for any $i\in[k]_0$, $\sum\limits_{\substack{j \in [k]_0, i \neq j}}|{\Upsilon}^{-1}(j)\cap \sigma(\widehat{t})| \leq 3{k}^2$. Thus, we have proved that ${\Upsilon}_{|F}$ is $(3{k}^2,k)$-unbreakable. $\diamond$

\medskip
\noindent Finally, we present our third claim.

\begin{claim}\label{claim:equivalent:3}
$\mu = \sum\limits_{F \in E(H)} {\mu}^F$, $ \sum\limits_{F \in E(H)} l_1^F \leq l_1$ and $\sum\limits_{F \in E(H)} l_2^F \leq l_2$.
\end{claim}
{\noindent\bf Proof.}
By the property of $(T,\beta)$ being a tree decomposition, for any two children $t_i$ and $t_j$ of $t$ in $T$, $\gamma(t_i) \cap \gamma(t_j) \subseteq \beta(t)$, and also from definition, $\sigma(t_i) \subseteq \beta(t)$ for any child $t_i$ of $t$. 
Now, note that $\mu = |\widehat{\Upsilon}_{\mathbf{v}}^{-1}(V_1)|$. Thus, to show that $\mu = \sum_{F \in E(H)} {\mu}^F$, it is sufficient to show that $|\widehat{\Upsilon}_{\mathbf{v}}^{-1}(V_1)| = \sum_{F \in E(H)} {\mu}^F$. However, keeping the above argument in mind, the claim that $|\widehat{\Upsilon}_{\mathbf{v}}^{-1}(V_1)| = \sum_{F \in E(H)} {\mu}^F$ directly follows from the satisfaction of the three following conditions. We remark that the satisfaction of these conditions is a direct consequence of the supposition that $\Upsilon$ be a witnessing coloring for $y[t, {\Upsilon}^{\sigma}, \mathbf{v},\mu,l_1,l_2]=1$, together with our definition of the values ${\mu}^F, l_1^F$ and $l_2^F$.
\begin{enumerate}
    \item For any type-1 hyperedge $F$, we have that ${\mu}^F = 1$ only if $\widehat{\Upsilon}_{\mathbf{v}}(F) = V_1$. In particular, $\sum\limits_{\substack{F\in E(H) \text{ of type-1}}} {\mu}^F = |\widehat{\Upsilon}_{\mathbf{v}}^{-1}(V_1) \cap \beta(t)|.$
    \item For any type-2 hyperedge $F$, ${\mu}^F = 0$. Thus, $\sum\limits_{\substack{F\in E(H) \text{ of type-2}}} {\mu}^F = 0$.
    \item For any type-3 hyperedge $F$, ${\mu}^F =|\widehat{\Upsilon}_{\mathbf{v}}^{-1}(V_1) \cap (\gamma(t_i) \setminus \sigma(t_i))|$. 
\end{enumerate}

\medskip

\noindent Similarly, let us observe that $|E(G[\widehat{\Upsilon}_{\mathbf{v}}^{-1}(V_1)])| \leq l_1$. Thus, to show that $\sum_{F \in E(H)} l_1^F \leq l_1$, it is sufficient to show that $\sum_{F \in E(H)} l_1^F \leq |E(G[\widehat{\Upsilon}_{\mathbf{v}}^{-1}(V_1)])|$. However, the latter inequality that directly follows from the satisfaction of all of the following conditions.

\begin{enumerate}
    \item For any type-1 hyperedge $F$, $l_1^F = 0$. Thus, $\sum\limits_{\substack{F\in E(H) \text{ of type-1}}} l_1^F = 0$.
    \item For any type-2 hyperedge $F =\{u,v\}$, ${l_1}^F = 1$ only if $\widehat{\Upsilon}_{\mathbf{v}}(u)  = \widehat{\Upsilon}_{\mathbf{v}}(v) = V_1$. In particular, $\sum\limits_{\substack{F\in E(H) \text{ of type-1}}} l_1^F = |E(G[\widehat{\Upsilon}_{\mathbf{v}}^{-1}(V_1)]) \cap E(G[\beta(t)])|$.
    \item For any type-3 hyperedge $F$, $|E(G[\widehat{\Upsilon}_{\mathbf{v}}^{-1}(V_1) \cap (\gamma(t_i)\setminus \sigma(t_i))])| \leq l_1^{F}$.
\end{enumerate}

\medskip

\noindent Symmetrically, $\sum_{F \in E(H)} l_2^F \leq l_2$. This concludes the proof of the claim. $\diamond$

\medskip

\noindent As we have proved Claims \ref{claim:equivalent:2} and \ref{claim:equivalent:3}, we derive that ${\Upsilon}|_{\beta(t)}$ is a witnessing coloring for {\ttfamily aHP}$[\mu,l_1,l_2]=1$. This concludes the proof of the lemma.
\end{proof}

Recall that we have argued that to prove Theorem \ref{thm:hp}, it is sufficient to show that the current set of values $y[t,{\Upsilon}^{\sigma},\mathbf{v},\mu,l_1,l_2]$ can be computed in time $2^{k^{\OO(1)}} n^{\OO(1)}$. Here, $n$ refers to $|V(G)|$. By Lemmas \ref{claim:semiequivalent} and \ref{claim:equivalent}, this set of values can be derived from the solution of $\hp(k_1,k_2,n, \eta, q, H,$ $\{f_F\}|_{F \in E(H)})$. Since $\hp(k_1,k_2,n, \eta, q, H,$ $\{f_F\}|_{F \in E(H)})$ is a favorable instance of \hp\ (by Lemma \ref{claim:favourable}), the algorithm given by Theorem \ref{thm:hp} solves it in time $2^{\OO(\min (k,q) \log (k+q))} d^{\OO(k^2)} |E(H)|^{\OO(1)}=2^{k^{\OO(1)}} n^{\OO(1)}$. % do not change this |E(H)| to m !!!
%This concludes the proof of Theorem~\ref{thm:hp}.

\section{Solving Favorable Instances of \hp}\label{sec:hp}

Recall the problem statement of \hpfull (\hp){} and the definition of a \emph{favourable} instance of \hp{} from Section \ref{sec:abcbjb}.
In this section, we prove Theorem \ref{thm:hp}.
For this purpose, let $HP(k_1,k_2,b,d,q,H,(f_F)_{F \in E(H)})$ be a favorable instance of \hp. We aim to show how to compute $\text{{\ttfamily aHP}}[\mu,l_1,l_2]$ in time $2^{\OO(\min (k,q) \log (k+q))} \allowbreak d^{\OO(k^2)} \allowbreak m^{\OO(1)}$ for an arbitrarily fixed choice of $0 \leq \mu \leq b$, $0 \leq l_1 \leq k_1$ and $0 \leq l_2 \leq k_2$. Since there are only $(b+1)(k_1+1)(k_2+1)$ choices for such $\mu, l_1$ and $l_2$, we would thus indeed derive the correctness of Theorem \ref{thm:hp}.

\subsection{Classifying Hyperedges}

We begin by analyzing the structure of the input instance $HP(k_1,k_2,b,d,q,H,(f_F)_{F \in E(H)})$ under the assumption that $\text{{\ttfamily aHP}}[\mu,l_1,l_2]=1$. Recall that $k =k_1 +k_2$. Then, by the property of global unbreakability, there exists a witnessing coloring $\Upsilon : V(H) \to [k]_0$ such that $\sum\limits_{\substack{j \in [k]_0, j\neq i}} |{\Upsilon}^{-1}(j)| \leq q$ for some index $i \in [k]_0$. Without loss of generality, suppose that $i=0$ is such an index, that is, $\sum_{j \in [k]}|{\Upsilon}^{-1}(j)| \leq q$.  In the forthcoming arguments, we aim to elucidate the behavior of the function that is the restriction of the witnessing coloring $\Upsilon$ with respect to each hyperedge of the hypergraph $H$. As we  see later, we may not be able to find the restriction of $\Upsilon$ on every hyperedge, but we would be able to assign a set of colorings to each hyperedge and prove that one of them is exactly the restriction of $\Upsilon$ to that hyperedge. We will then use this information together with dynamic programming procedures to compute {\ttfamily aHP}$[\mu,l_1,l_2]$. The difficulty lies in the fact that if the set of colorings (with the above-mentioned property) that we would like to have with respect to each hyperedge is arbitrary, the efficiency of our dynamic programming based procedures would not be guaranteed. More precisely, at any given point of the computation, when we choose some  coloring for a specific hyperedge using the respective set of colorings of that hyperedge, we would like to be able to automatically assume that this coloring together with all previously chosen colorings should form one coherent coloring that is compatible with a global witnessing coloring. In order to achieve such a property, we perform several phases of color coding of the hypergraph (here, these phases of color coding are hidden under a layer of derandomization tools). These phases would exploit the properties of a favourable instance, and eventually highlight a ``nice'' structure that would help us achieve our goal.

To proceed with the implementation of the above-mentioned idea, we first categorize the hyperedges of $H$ into the following types, based on the witnessing coloring $\Upsilon$. In this context, we remind that the notation $f(A') =b$ indicates that for all $a \in A'$, it holds that $f(a)=b$ (see Section \ref{sec:prelims}). 
\begin{itemize}
\item Let $E_b = \{F\in E(H) : \Upsilon(F) = 0\}$. Here, `b' stand for {\em big}.
\item For each $i \in [k]$, let $E_{s_i}=\{F\in E(H) : \Upsilon(F)=i \}$. Here, `s' stands for {\em small}.
\item Let $E_m =\{F\in E(H) : \text{ there exist } u,v \in F \text{ such that } \Upsilon(u) \neq \Upsilon(v) \}$. Here, `m' stands for {\em multichromatic}.
\end{itemize}

Observe that each hyperedge $F \in E(H)$ belongs to exactly one of the sets $E_b, E_m, E_{s_1}, \ldots,$ $E_{s_k}$. Furthermore, let $E_{s_i}'$ denote the edge set of some arbitrary spanning forest of the hypergraph on the vertex set $V(H)$ and the edge set $E_{s_i}$. Let $E_s = \bigcup_{i \in [k]} E_{s_i}'$ denote the union of these edge sets. Since we are working with a favourable instance of \hp, we would see (in Lemmas \ref{claim:hp:esbounded} and \ref{claim:hp:embounded}) that the sizes of the sets $E_s$ and $E_m$ can be upper bounded by $q$ and $k$, respectively. We exploit these bounds to highlight the hyperedges in $E_m$ and $E_s$ (Lemma \ref{claim:hp:goodassignments}) efficiently. In addition to this, as we shall see in Lemma \ref{claim:hp:alpha}, the total number of possible restrictions of $\Upsilon$ on any hyperedge can also bounded effectively. Thus, we can not only highlight the hyperedges in $E_m$ and $E_s$, but we can also guess the restrictions of $\Upsilon$ to these hyperedges. We remark that since we aim to solve a favourable instance of \hp{} in time that is proportional to a single exponential function of $q$, we do not guess the restriction of $\Upsilon$ to the hyperedges of $E_s$ straightaway (as $|E_s| \leq q$ from Lemma \ref{claim:hp:esbounded}). The proof of Lemma \ref{claim:hp:goodassignments} would capture the idea of the performance of highlighting and guessing. As one would expect, this highlighting does concludes our arguments, as it does not just highlight the hyperedges in $E_m$ and $E_s$, but also some hyperedges from $E_b$. We deal with the inherent challenges of handling such a ``messy picture'' later in our proof.
%In order to make the best use of such a highlighting, appropriate segregation is required to be done. What comes to rescue here is an appropriate highlighting procedure
%\begin{obs}$E(H) = E_b \uplus E_m \uplus \biguplus\limits_{i \in [k]} E_{s_i}  $.\end{obs}

%A hyperedge $F \in E(H)$ is called a {\em big} hyperedge if $F \in E_b$. It is called an {\em $i$-colored small} hyperedge in $F \in E_{s_i}$ and a {\em multichromatic} hyperedge if $F \in E_m$.

\begin{lemma}\label{claim:hp:esbounded}
$|E_s| \leq q$.
\end{lemma}
\begin{proof}
Recall that for each $i \in [k]$, we defined $E_{s_i}'$ as the edge set of a spanning forest of the hypergraph with the vertex set $V(H)$ and the edge set $E_{s_i}$. Hence, by this definition,  $|E_{s_i}'| \leq |{\Upsilon}^{-1}(i)|$. Now, recall that since $\Upsilon$ witnesses the global unbreakability property, we assumed w.l.o.g.~that $\sum_{i \in [k]}|{\Upsilon}^{-1}(i)| \leq q$. We thus have that $\sum_{i \in [k]} |E_{s_i}'| \leq q$. Therefore, $|E_s| \leq q$.
\end{proof}

\begin{lemma}\label{claim:hp:embounded}
$|E_m| \leq k$.
\end{lemma}
\begin{proof}
Since $\text{{\ttfamily aHP}}[\mu,l_1,l_2]=1$, for all $F \in E(H)$ there exist ${\mu}^F$, $l_{1}^F$ and $l_{2}^F$ such that $f_F({\Upsilon}_{|F}, {\mu}^F, l_{1}^F, l_{2}^F) =1$. Hence, the connectivity property implies that for each $F \in E_m$, we have that $l_{1}^F + l_{2}^F \geq 1$. However, $\sum_{F \in E(H)} l_1^F + l_2^F \leq l_1 + l_2 \leq k_1 + k_2 =k$. Thus, $|E_m| \leq k$.
\end{proof}

\subsection{Introducing Good Assignments}

Let us first note that by Lemma \ref{lemma:numOfUnbreakableFunctions}, for any hyperedge $F\in E(H)$, the number of $(3{k}^{2},k)$-unbreakable functions (that we call  $(3{k}^{2},k)$-unbreakable colorings) from $F$ to $[k]_0$ is at most $\alpha =  \sum\limits_{l=1}^{3k^2} {d \choose l} \cdot (3k^2)^k \cdot (k+1) = d^{\OO(k^2)}$. For each hyperedge $F$, let us arbitrarily order all possible $(3{k}^{2},k)$-unbreakable colorings. For each $i \in [\alpha]$, let ${\lambda}_{F,i}$ denote the $i$-th such coloring. If for an heperedge $F$, the number of such colorings is strictly smaller than $\alpha$, then we extend its list of possible colorings to be of size $\alpha$ by letting some colorings be present multiple times. Thus, for each $F \in E(H)$ and $i \in [\alpha]$, we ensure ${\lambda}_{F,i}$ is well-defined.

\begin{lemma}\label{claim:hp:alpha}
For any $F \in E(H)$, there exists $i \in [\alpha]$ such that $\Upsilon|_{F} = {\lambda}_{F,i}$.
\end{lemma}
\begin{proof}
Since $\Upsilon$ is a witnessing coloring for $\text{{\ttfamily aHP}}[\mu,l_1,l_2]=1$, we have that for any $F \in E(H)$ there exist ${\mu}^F$, $l_{1}^F$ and $l_{2}^F$ such that $f_F({\Upsilon}_{|F}, {\mu}^F, l_{1}^F, l_{2}^F) =1$. Since the property of local unbreakability is then enforced by the definition of $f_F$, we have that ${\Upsilon}|_{F}$ is a $(3{k}^{2},k)$-unbreakable coloring. Since $\{{\lambda}_{F,1}, \ldots, {\lambda}_{F, \alpha}\}$ is contains all possible $(3{k}^2,k)$-unbreakable colorings from $F$ to $[\alpha]$, there exists $i \in [\alpha]$ such that ${\Upsilon}_{|F} = {\lambda}_{F,i}$.
\end{proof}

Here, we are interested in assignments that are functions associating each hyperedge $F\in E(H)$ with a coloring ${\lambda}_{F,i}$. Let us proceed by defining which assignments would be useful for us to have at hand.

\begin{definition}\label{def:goodAssignment}
An assignment $p : E(H) \to [\alpha]_0$ is said to be a \emph{good} assignment if the following conditions hold.
\begin{enumerate}
    \item For all $F \in E_s$, $p(F)=0$.
    \item For all $F \in E_m$, $p(F) =i >  0$ and ${\Upsilon}|_{F} = {\lambda}_{F,i}$.
\end{enumerate}
\end{definition}

To employ coloring coding, we first mention the required derandomization tools.

\begin{proposition}[Lemma 1.1, \cite{chitnis12}]\label{lemma:derand1}
Given a set $U$ of size $n$ and $y,z \in [n]_0$, we can construct in time $\OO(2^{\OO(\min (y,z) \log (y+z))} n \log n)$ a family $\mathcal{F}$ of at most $\OO(2^{\OO(\min (y,z) \log (y+z))} \log n)$ subsets of $U$, such that the following holds: for all sets $Y,Z \subseteq U$ such that $Y \cap Z = \emptyset$, $|Y| \leq y $ and $|Z| \leq z$, there exists a set $S \in \mathcal{F}$ with $Y \subseteq S$ and $Z \cap S = \emptyset$.
\end{proposition}

\begin{definition}[$(N,r)$-perfect family]
An $(N,r)$-perfect family is a family of functions from $[N]$ to $[r]$, such that for any subset $X \subseteq [N]$ of size $r$, there exists a function in the family that is injective on $X$.
\end{definition}

\begin{proposition}[\cite{naor95}]\label{lemma:derand2}
An $(N,r)$-perfect family of size $\OO(e^r r^{\OO(\log r)} \log N)$ can be computed in time $\OO(e^r r^{\OO(\log r)} \allowbreak N \log N)$.
\end{proposition}

We are now ready to present our color coding phases.

\begin{lemma}\label{claim:hp:goodassignments}
There exists a set $\mathcal{A}$ of assignments from $E(H)$ to $[\alpha]_0$, such that $|\mathcal{A}| \leq 2^{\OO(\min (k,q) \log (k+q))} \cdot d^{\OO(k^2)}  \cdot {\log}^2 |E(H)|$ and there exists a good assignment in $\mathcal{A}$. Moreover, such a set $\mathcal{A}$ is computable in time $\OO(2^{\OO(\min (k,q) \log (k+q))} \cdot d^{\OO(k^2)} \cdot |E(H)|^{\OO(1)})$.
\end{lemma}

\begin{proof}
%\textbf{Description of the family $\mathcal{A}$}
We start by defining three several families, which would guide us through the construction of $\cal A$. For $U= E(H)$, $y = k$ and $z = q$, let $\mathcal{F}=\{S_1, \ldots, S_{\nu}\}$ be the family of size $\nu = 2^{\OO(\min (k,q) \log (k+q))} \log |E(H)|$ obtained by calling the algorithm of Proposition \ref{lemma:derand1}. For each $j \in [\nu]$, let ${\mathcal{P}}_j$ be a $(|S_j|,k)$-perfect family of size at most $\zeta = e^k k^{\OO(\log k)} \log |S_j| = e^k k^{\OO(\log k)} \log |E(H)|$ computed by the algorithm of Proposition \ref{lemma:derand2}. Let $\mathcal{Q}$ be the family of all possible functions from $[k]$ to $[\alpha]$. Observe that $|\mathcal{Q}| = {\alpha}^k$.

For each set $S_j \in \mathcal{F}$, function $\kappa \in {\mathcal{P}}_j$ and function ${\kappa}_0 \in \mathcal{Q}$, let $p[S_j, \kappa , {\kappa}_0 ] : E(H) \to [\alpha]_0$ be defined as follows.
 \[
    p[S_j, \kappa, \kappa_0](F) =
    \begin{cases}
     0,              & \text{if } F \in S_j \\     
     {\kappa}_{0}(\kappa(F))  & \text{otherwise}
    \end{cases}
    \]
%\textbf{$\mathcal{A}$ satisfies the desired properties}

Let $\mathcal{A}=\{p[S_j, \kappa , {\kappa}_0] : S_j \in \mathcal{F}, \kappa \in {\mathcal{P}}_j, {\kappa}_0 \in \mathcal{Q}\}$. We claim that there exists a good assignment in $\mathcal{A}$. Since $|E_m| \leq k$ (from Lemma \ref{claim:hp:embounded}) and $|E_s| \leq q$ (from Lemma \ref{claim:hp:esbounded}), from Proposition \ref{lemma:derand1} there exists $S_j \in \mathcal{F}$ such that $E_m \subseteq S_i$ and $E_s \cap S_j = \emptyset$. By Proposition \ref{lemma:derand2}, there exists a function $\kappa \in {\mathcal{P}}_j$ which is injective on $E_m$. Let $E_m = \{F_1, \ldots, F_c\}$ where $c \leq k$. Without loss of generality, $\kappa(E_y) = y$ for all $y \in [c]$. Since $\mathcal{Q}$ contains all possible functions from $[k]$ to $[\alpha]$, and for each $F \in E_m$ there exists $i \in [\alpha]$ such that ${\Upsilon}|_{F} = {\lambda}_{F,i}$ (from Lemma \ref{claim:hp:alpha}), there exists ${\kappa}_0  \in \mathcal{Q}$ such that for each $F \in E_m$, ${\Upsilon}|_{F} = {\lambda}_{F, {\kappa}_0(\kappa(F))}$. Moreover, since $E_s \cap S_j = \emptyset$, we have that $p[S_j, \kappa, {\kappa}_0]=0$. Thus, $p[S_j, \kappa, {\kappa}_0] \in \mathcal{A}$ is a good assignment.

%\textbf{Size of $\mathcal{A}$}
Recall that $\alpha= d^{\OO(k^2)}$. Now, as we have upper bounded $\nu$ and $\zeta$, we observe that $|\mathcal{A}| \leq \nu \zeta {\alpha}^{k} = 2^{\OO(\min (k,q) \log (k+q))} e^k k^{\OO(\log k)}d^{\OO(k^2)} {\log}^2 |E(H)|$. Thus, the size of $\mathcal{A}$ is upper bounded by $2^{\OO(\min (k,q) \log (k+q))} \cdot d^{\OO(k^2)}  \cdot {\log}^2 |E(H)|$. 

%\textbf{Time taken to compute $\mathcal{A}$}
The time taken to compute $\mathcal{A}$ is proportional to the time taken to compute $\mathcal{F}, \mathcal{P}_j$ for each $j \in \{\nu\}$ and $\mathcal{Q}$. By Propositions \ref{lemma:derand1} and \ref{lemma:derand2}, we thus derive that the running time is upper bounded $\OO(2^{\OO(\min (k,q) \log (k+q))} \cdot d^{\OO(k^2)} \cdot |E(H)|^{\OO(1)})$.
\end{proof}

The algorithm we design to compute $\text{{\ttfamily aHP}}[\mu,l_1,l_2]$ first constructs the set $\mathcal{A}$ of Lemma \ref{claim:hp:goodassignments}. Observe that this computation can be done regardless of whether $\text{{\ttfamily aHP}}[\mu,l_1,l_2]=1$ or $\text{{\ttfamily aHP}}[\mu,l_1,l_2]=0$. (We only use the supposition that $\text{{\ttfamily aHP}}[\mu,l_1,l_2]=1$ to analyze structural properties of an input instance satisfying this condition.) Next, the algorithm branches on all possible assignments in $\mathcal{A}$. By Lemma \ref{claim:hp:goodassignments}, assuming that $\text{{\ttfamily aHP}}[\mu,l_1,l_2]=1$, we know that there exists at least one assignment from $E(H)$ to $[\alpha]_0$ that is good. Henceforth, we assume that we currently consider a branch that corresponds to a good assignment, denoted by $p : E(H) \to [\alpha]_0$. Thus, we would like to show that we correctly determine at the current branch that $\text{{\ttfamily aHP}}[\mu,l_1,l_2]=1$. (If it were the case that $\text{{\ttfamily aHP}}[\mu,l_1,l_2]=0$, it would also be clear from our arguments that we would not determine that $\text{{\ttfamily aHP}}[\mu,l_1,l_2]=1$, which would overall imply that no branch determines that this condition holds, and hence we would eventually decide that $\text{{\ttfamily aHP}}[\mu,l_1,l_2]=0$.)

\subsection{Associating the Graph $L_p$ with an Assignment $p$}

For our assignment $p : E(H) \to [\alpha]_0$, let us now construct an undirected simple graph $L_p$ with $V(L_p) = V(H)$. For each $F \in E(H)$ such that $p(F) =0$, make $F$ a clique in $L_p$. We say that the edges of this clique are the edges that \emph{correspond} to the hyperedge $F$. For any $F \in E(H)$ such that $p(F) =i > 0$, for each $j \in [k]_0$, make the set ${{\lambda}_{F,i}}^{-1}(j)$ a clique in $L_p$. We say that the edges of all such cliques are the edges that \emph{correspond} to the hyperedge $F$. Since we want $L_p$ to be a simple graph, between any two vertices of $L_p$ we retain at most one copy of the edge between them (if one exists). If a deleted copies of some edge $e$ in $L_p$ corresponds to some hyperedge $F$, then in the simple graph the retained copy of that edge $e$ is the one that is said to correspond to that hyperedge $F$ (even if we originally added the retained copy of $e$ due to a different hyperedge). Note that it may thus be the case that one edge in $L_p$ corresponds to to seversal hyperedges in $E(H)$.

We proceed by analyzing the connected components of $L_p$. Informally, we first argue that every connected component behaves as a single unit with respect to $\Upsilon$.

\begin{lemma}\label{claim:hp:prop1}
Let $D$ be any connected component of $L_p$. Then, $\Upsilon(D) = i$ for some $i\in[k]_0$, that is, all the vertices in $D$ are assigned the same color by $\Upsilon$.
\end{lemma}
\begin{proof}
For any $\mathcal{F} \subseteq E(H)$, let $L_p[\mathcal{F}]$ be the simple graph on the same vertex set as $L_p$, whose edge set contains only those edges of $L_p$ that correspond to some hyperedge in $\mathcal{F}$. Observe that $L_p[E(H)] =L_p$. Moreover, observe that if a set of vertices is connected in $L_p[\mathcal{F}]$ then it is also connected in $L_p[{\mathcal{F}}']$ for any ${\mathcal{F}}' \supseteq \mathcal{F}$.

Let $E(H) = \{F_1, \ldots, F_r\}$. Moreover, for any $j\in[r]$, denote ${\mathcal{F}}_j = \bigcup\limits_{c=1}^{j} F_c$. Let us prove by induction on $j$ that for each component $D$ of $L_p[{\mathcal{F}}_j]$, we have that $\Upsilon(D) = i$ for some $i \in [k]_0$. The proof of this claim would conclude the proof of the lemma, as by setting $j=r$, we thus derive that for each component $D$ of $L_p[{\mathcal{F}}_r] = L_p$, we have that $\Upsilon(D)=i$ for some $i \in [k]_0$. Hence, we next focus only on the proof of the claim.

To prove the base case, where $j=1$, consider the graph $L_p[{\mathcal{F}}_1]$. 
If $F_1 \not \in E_m$, then $\Upsilon(F_1) = i$ for some $i \in [k]_0$ (by the definition of $E_m$). Hence, for each connected component $D$ of $L_p[{\mathcal{F}}_1]$, $\Upsilon(D)=i$ for some $i \in [k]_0$. Otherwise, $F_1 \in E_m$. In this case, let $p(F_1) =s > 0$. Since $p$ is a good assignment, ${\lambda}_{F_1 , s} = {\Upsilon}|_{{F}_{1}}$. Since each component $D$ of $L_p[{\mathcal{F}}_1]$ is either an isolated vertex or ${\lambda}_{F_1 , s}^{-1}(i)$ for some $ i \in [k]_0$, we conclude that $\Upsilon(D)=i$ for some $i \in [k]_0$.

%If $p(F_1) =0$, then there is only one connected component, say $D$, in $L_p[{\mathcal{F}}_1]$. Since $p(F_1)=0$, $F_1 \not \in E_m$ because $p$ is a good assignment. Hence $\Upsilon(F) = \Upsilon(D) = i$ for some $i \in \{0. \ldots, a\}$. If $p(F_1) = s >0$ and $F_1 \in E_m$, then ${\Upsilon}_{|F_1} = {\lambda}_{F_1, s}$. Since any connected component $D$ of $L_p[{\mathcal{F}}_1]$ is such that $D = {\lambda}_{F_1, s}^{-1} (i)$ for some $i \in \{0, \ldots,p\}$, therefore, $\Upsilon(D)=i$ for for some $i \in \{0, \ldots, p\}$. Also, if  $p(F_1) = s >0$ and $F_1 \not \in E_m$, then 

We now suppose that $j\geq 2$. By induction hypothesis, for each connected component $D$ of $L_p[{\mathcal{F}}_{j-1}]$, we have that $\Upsilon(D) = i$ for some $i \in [k]_0$. Let us now examine the graph $L_p[{\mathcal{F}}_j]$ and the hyperedge $F_j$. Note that $F_j={\cal F}_j\setminus{\cal F}_{j-1}$.
If $F_j \not \in E_m$, then $\Upsilon(F_j) = i$ for some $i \in  [k]_0$ (from the definition of $E_m$). Let $\mathcal{D}$ be the collection of every connected components of $L_p[{\mathcal{F}}_{j-1}]$ which intersects $F_j$. Then, the definition of $L_p$ and the inductive hypothesis directly imply that $\Upsilon(\bigcup \mathcal{D}) = i$ for some $i \in [k]_0$. Thus, by the inductive hypothesis, for each connected component $D$ of $L_p[{\mathcal{F}}_{j}]$, we have that $\Upsilon(D)=i$ for some $i \in [k]_0$. Otherwise, $F_j \in E_m$. Then, denote $p(F_1) =s > 0$. Since $p$ is a good assignment, ${\lambda}_{F_1 , s} = {\Upsilon}|_{{F}_{1}}$. For each $i \in [k]_0$, let ${\mathcal{D}}_i$ be the collection of all connected components of $L_p[{\mathcal{F}}_{j-1}]$ that intersect ${\lambda}_{F_j, s}^{-1}(i)$. Then, the definition of $L_p$ and the inductive hypothesis directly imply $\Upsilon({\mathcal{D}}_{i}) = i$. Hence, by the inductive hypothesis, for each connected component $D$ of $L_p[{\mathcal{F}}_j]$, we have that $\Upsilon(D) = i$ for some $i \in [k]_0$. 
\end{proof}

Roughly speaking, we now argue that hyperedges crossing several different components, where to at least one of them $\Upsilon$ assigns some $i>0$, should belong to $E_m$.

\begin{lemma}\label{claim:hp:prop2}
Let $D$ be any connected component of $L_p$ such that $\Upsilon(D) = i > 0$ for some $i\in[k]$. For any $F \in E(H)$ such that $F \cap D \neq \emptyset$ and $F \setminus D \neq \emptyset$, then $F \in E_m$.
\end{lemma}
\begin{proof}
Suppose that the statement is false, that is, there exists $F \in E(H)\setminus E_m$ such that $F \cap D \neq \emptyset$ and $ F\setminus D \neq \emptyset$. Since $F \notin E_m$, $F \cap D \neq \emptyset$ and $\Upsilon(D) >0$, there exists $j \in [k]$ such that $F \in E_{s_j}$. Since $F \cap D \neq \emptyset$ and $\Upsilon(D) =i$, we have that $j=i$, that is, $F \in E_{s_i}$. Consider any spanning forest $E_{s_i}'$ of the hypergraph with vertex set $V(H)$ and edge set $E_{s_i}$. Observe that by the definition of $L_p$, for any spanning forest $E_{s_i}$, all vertices of $F$ are present in some single tree of that spanning forest. Therefore, there exists some $F' \in E_{s_i}'$, where $F'$ could be the same as the hyperedge $F$, such that the vertices of $F'$ form a superset of $F$. Since $p$ is a good assignment, $p(F') =0$. Thus, the definition of $L_p$ implies that all the vertices of $F$ belong to the same connected component, which contradicts that $F \setminus D \neq \emptyset$.
\end{proof}

\subsection{Rules to Modify a Good Assignment}

We first modify the good assignment $p$ by applying the following two rules exhaustively, prioritizing Rule 1 over Rule 2. Note that whenever we change $p$, we update $L_p$ accordingly. 

\begin{tcolorbox}
\noindent\textbf{Rule 1}: If there exist a connected component $D$ of $L_p$ and a hyperedge $F \in E(H)$ such that $F \subseteq D$ and $p(F) >0$, then update $p(F)=0$.
\end{tcolorbox}

\begin{tcolorbox}
\noindent\textbf{Rule 2}: If there exist a connected component $D$ of $L_p$, vertices $v_1,v_2 \in D$ ($v_1$ could be equal to $v_2$) and hyperedges $F_1, F_2 \in E(H)$ ($F_1$ could be equal to $F_2$) such that $F_1 \cap D \neq \emptyset$, $F_2 \cap D \neq \emptyset$, $F_1 \setminus D \neq \emptyset$, $F_2 \setminus D \neq \emptyset$, $p(F_1) = i >0$, $p(F_2)=j > 0$, ${\lambda}_{{F}_{1}, i}(v_1) \in [k]$ and ${\lambda}_{{F}_{2},j}(v_2)=0$, then update $p(F_1)=0$.
\end{tcolorbox}

\begin{lemma}\label{claim:hp:modifyp}
After any application of Rule 1 and Rule 2, $p$ remains a good assignment.
\end{lemma}

\begin{proof}
Let us first prove that if $p$ was a good assignment, then after the application of Rule 1, the modified $p$ is still a good assignment. From Lemma \ref{claim:hp:prop1}, $\Upsilon(D)=i$ for some $i \in [k]_0$. Thus, if $F\subseteq D$, then $F \not \in E_m$. Hence, when we redefine $p(F)=0$, $p$ remains a good assignment.

Let us now prove that if $p$ was a good assignment, then after the application of Rule 2, the modified $p$ is still a good assignment. To prove this, it is enough to prove that $F_1 \notin E_m$. Suppose, for the sake of contradiction, that $F_1 \in E_m$. Since $p$ is a good assignment, ${\lambda}_{F_1 , i} = {\Upsilon}|_{F}$. Denote ${\lambda}_{F_1 , i }(v_1) = c$, where $c \in [k]$. Since $v_1 \in D$ and ${\lambda}_{F_1 , i }(v_1) = c > 0$, from Lemma \ref{claim:hp:prop1}, $\Upsilon(D)=c > 0$. From Lemma \ref{claim:hp:prop2}, $F_2 \in E_m$. Again, since $p$ is a good assignment, ${\lambda}_{F_2 , j} = {\Upsilon}|_{F}$. Since ${\lambda}_{F_2 , j }(v_2) = 0$ and $v_2 \in D$, this implies that $\Upsilon(D) =0$, which is a contradiction. 
\end{proof}

For each connected component $D$ of $L_p$, let us now define a \emph{label set} $L(D) \subseteq [k]_0$ as follows. For any $i \in [k]_0$, we insert $i$ into $L(D)$ if and only if there exists $F \in E(H)$ such that $F \cap D \neq \emptyset$, $p(F) =j >0$ and ${\lambda}_{F,j}(F \cap D) = i$. Observe that $L(D)$ could be empty.

Let us now turn to analyze the labels sets we have just defined.

\begin{lemma}\label{lemma:hp:empty}
Let $D$ be connected component of $L_p$ such that $L(D) = \emptyset$. Then, for any $F \in E(H)$ such that $F \cap D \neq \emptyset$, $F \setminus D = \emptyset$.
%the hypergraph induced by the vertices of $D$ is an isolated component in $H$. 
\end{lemma}
\begin{proof}
Observe that if there exists $F \in E(H)$ such that $p(F) > 0$ and $F \cap D \neq \emptyset$, then $|L(D)| \geq 1$. Therefore, if $L(D) = \emptyset$, then for all $F \in E(H)$ such that $F \cap D \neq \emptyset$, we have that $p(F) =0$. Thus, from the construction of $L_p$, we have that $F \setminus D = \emptyset$.
\end{proof}

\begin{lemma}\label{claim:hp:multilabel}
For any connected component $D$ of $L_p$, if $\Upsilon(D) = i >0$, then either $L(D) = \emptyset$ or $L(D) = \{i\}$.
\end{lemma}

\begin{proof}
Suppose that $L(D) \neq \emptyset$. Then, there exists $F \in E(H)$ such that $F \cap D \neq \emptyset$ and  $p(F) = j > 0$. Let ${\lambda}_{F,j}(F \cap D) = s$. We will now show that $s=i$.
%We will now show that for any $F \in E(H)$ such that $F \cap D \neq \emptyset$, $p(F) = j > 0$ for some $j \in $ and ${\lambda}_{F,j}(F \cap D) = s$ for some $s \in \{1, \ldots, \alpha\}$, $s = i$.
First of all, let us argue that $F \setminus D \neq \emptyset$. Indeed, if $F \setminus D = \emptyset$, then $F \subseteq D$. In this case, since $p$ is a good assignment, where Rule 1 has been exhaustively applied, $p(F)$ should be equal to $0$, which is a contradiction. Thus, since $\Upsilon(D) = i > 0$, $F \cap D \neq \emptyset$ and $F \setminus D \neq \emptyset$, from Lemma \ref{claim:hp:prop2}, we have that $F \in E_m$. Then, since $p$ is a good assignment, ${\lambda}_{F,j}(F \cap D) = {\Upsilon}|_{F}$. Since $\Upsilon(D)=i$, we derive that indeed ${\lambda}_{F,j}(F \cap D)=i$. Thus, $L(D) = \{i\}$.
\end{proof}

By Lemma \ref{claim:hp:multilabel}, we have that if for a connected component $D$ of $L_p$, either $L(D) = \{0\}$ or $|L(D)| \geq 2$, then $\Upsilon(D) = 0$.

\begin{lemma} \label{claim:hp:ld0}
If $D$ is a connected component of $L_p$ such that $L(D) = \{l_d\}$, then either $\Upsilon(D) = l_d$ or $\Upsilon(D) =0$. 
\end{lemma}

\begin{proof}
 Since $L(D)=\{l_d\}$, there exists $F \in E(H)$ such that $p(F) = i > 0$, $F \cap D \neq \emptyset$ and ${\lambda}_{F,i}(F \cap D) = l_d$. Denote $\Upsilon(D)=j$, and suppose that $j \neq 0$, else we are done. Since $j \neq 0$, from Lemma \ref{claim:hp:prop2} we have that $F \in E_m$. Then, since $p$ is a good assignment, ${\lambda}_{F,i} = {\Upsilon}|_{F}$. Finally, since all the vertices of $D$ are assigned the same color by $\Upsilon$ (by Lemma \ref{claim:hp:prop1}), we have that $\Upsilon(D) = l_d$.
\end{proof}

For a connected component $D$ of $L_p$ such that $|L(D)| \geq 2$, let us redefine the label set of $D$ to be $L(D) =\{0\}$. Now, for any connected component $D$ of $L_p$, $|L(D)| \leq 1$. Moreover, if $L(D) = \{0\}$, then $\Upsilon(D) =0$ (by Lemma \ref{claim:hp:multilabel}). 
We call a connected  component $D$ of $L_p$ such that $L(D)=\{0\}$ a {\em $0$-component}.
%\begin{lemma}Let $D$ be some connected component of $L_p$. If $L(D) = \{0\}$, then $\Upsilon(D) = 0$.\end{lemma}

%\begin{lemma}\label{claim:hp:labelsmall}If $\Upsilon(D)=i$, for some $i \in [k]$, then $|L(D)| \leq 1$.\end{lemma}

%Thus, from claim \ref{claim:hp:labelsmall}, if $|L(D)| \geq 2$, then $\Upsilon(D) =0$. 

%Let us now redefine the label sets of all those connected components of $D$ that have at least $2-$sized label sets. Formally, let $D$ be some connected component of $L_p$. If $|L(D)| \geq 2$, then redefine $L(D) = \{0\}$. 
%Let us call a connected $D$ of $L_p$ a $0$-component, if $|L(D)| \geq 2$. Let us now redefine the label sets of the $0$-components. If $D$ is a $0$-component of $L_p$, then set $L(D) =\{0\}$. 

%\begin{obs}For any connected component $D$ of $L_p$, $|L(D)| \leq 1$. Also, if $L(D) = \{0\}$ and $\Upsilon(D) =0$.\end{obs}

Let us continue modifying the good assignment $p$, now with the following rule. Again, whenever we modify $p$, we update $L_p$ accordingly.

\begin{tcolorbox}
\noindent\textbf{Rule 3:} If there exist $F \in E(H)$ and two distinct $0$-connected components of $L_p$, $D_1$ and $D_2$, such that $F \cap D_1 \neq \emptyset$ and $F \cap D_2 \neq \emptyset$, then update $p(F)=0$.
\end{tcolorbox}

\begin{lemma}\label{claim:hp:modifyp1}
The assignment resulting by applying Rule 3 to $p$ is a good assignment.
\end{lemma}

\begin{proof}
 To prove the lemma, it is sufficient to show that $F \notin E_m$. Suppose that this claim is false, that is, $F \in E_m$ and hence after the update, we obtain an assignment that is not good. Since (the original) $p$ is a good assignment, we have that $p(F) = i > 0$ such that ${\lambda}_{F,i} = {\Upsilon}|_{F}$. Since $D_1$ and $D_2$ are different connected components of $L_p$, $(F \cap D_1) \subseteq {\lambda}_{F,i}^{-1}(j_1)$, $(F \cap D_2) \subseteq {\lambda}_{F,i}^{-1}(j_2)$ and $j_1 \neq j_2$. However, since $D_1$ and $D_2$ are $0$-components of $L_p$, $\Upsilon(D_1) = 0$ and $\Upsilon(D_2) = 0$. This contradicts that ${\lambda}_{F,i} = {\Upsilon}|_{F}$. Therefore, $F \notin E_m$.
\end{proof}

%%%%%%%%%%%%%%%%%%%%%%%%%%%%%%%%%%%%%%%%%%%%%%
To further analyze $0$-components, define $B$ as the set containing every vertex $v \in V(H)$ such that $\Upsilon(v) =0$ and there exists $F\in E_m$ that is incident to $v$. 
%The set $B$ basically contains all the vertices of the hypergraph

\begin{lemma} \label{claim:hp:bigvertices}
Let $D$ be a connected component of $L_p$ containing a vertex $v\in B$. Then, $D$ is a $0$-component.
\end{lemma}
\begin{proof}
From the definition of the set $B$, there exists $F \in E_m$ such that $v \in F$. Since $p$ is a good assignment, $p(F) = i > 0$ such that ${\lambda}_{F,i} = {\Upsilon}|_{F}$. Since $\Upsilon(v) =0$, $v \in F$ and $v \in D$, we have that ${\lambda}_{F,i}(F \cap D) = 0$. Hence, $0 \in L(D)$. Therefore, by Lemma \ref{claim:hp:multilabel}, we conclude that $D$ is a $0$-component of $L_p$.
\end{proof}
%%%%%%%%%%%%%%%%%%%%%%%%%%%%%%%%%%%%%%%%%%%%%%%
%Construction of $L_{p}^{*}$\begin{abstract}

\subsection{Constructing a Supergraph $L^*_p$ of $L_p$}\label{sec:supergraph}

Let us now construct another simple undirected graph ${L}^{*}_p$, which is a supergraph of $L_p$ with the same vertex set as of $L_p$ and the following additional edges. If there exist $F \in E(H)$ and two distinct connected components of $L_p$, $D_1$ and $D_2$, such that $F \cap D_1 \neq \emptyset$, $F \cap D_2 \neq \emptyset$, $L(D_1) \neq \{0\}$ and $L(D_2) \neq \{0\}$, then insert an edge between some vertex of $D_1$ and some vertex of $D_2$ into ${L}^{*}_p$.
Clearly, any connected component $D$ of $L_p$ is contained in some connected component of $L^{*}_{p}$. This leads us to the following definition.
\begin{definition}
Give a connected component $D^{*}$ of \lstarp, we say that a connected component $D$ of $L_p$ is a \emph{constituent} of \dstar{} if $D \subseteq {D}^{*}$.
\end{definition}

% 0-component of $L_{p}^{*}$
A component ${D}^{*}$ of \lstarp{} is called a $0$-component of \lstarp{} if it has only one constituent component and that constituent component is a $0$-component in $L_p$. We now proceed to analyze the new graph $L^*_p$.

\begin{lemma}\label{claim:hp:emptylabel}
Let \dstar{} be some connected component of \lstarp\ that has a constituent component $D$  such that $L(D) = \emptyset$. Then, $D$ is the only constituent component of \dstar, that is, ${D}^{*}=D$.
\end{lemma}
\begin{proof}
By Lemma \ref{lemma:hp:empty}, for any $F \in E(H)$ such that $F \cap D \neq \emptyset$,  we have that $F \setminus D = \emptyset$. 
%Observe that, for any connected component $D$ of $L_p$, if there exists $F \in E(H)$ such that $p(F) > 0$ and $F \cap D \neq \emptyset$, then $|L(D)| \geq 1$. Therefore, if $L(D) = \emptyset$ then for all $F \in E(H)$ such that $F \cap D \neq \emptyset$, $p(F) =0$. Therefore, $F \setminus D = \emptyset$. 
Thus, by the construction of $L_{p}^{*}$, it holds that $D^{*} = D$.
\end{proof}

\begin{lemma}\label{claim:hp:disjointhyperedges}
For any $F \in E(H)$, either $F \subseteq D^{*}$ for some connected component $D^{*}$ of $L^{*}_p$, or $F$ intersects exactly two connected components of $L^*_p$, a 0-connected component $D^{*}_{1}$ of $L^*_p$ and a non $0$-connected component $D^{*}_{2}$ of $L^{*}_p$.
\end{lemma}
\begin{proof}
Suppose that there exists  $F \in E(H)$ such that for any connected component $D^{*}$ of $L_{p}^{*}$, $F \not  \subseteq D^{*}$, else we are done. First, observe that in this case $p(F) >0$, as otherwise $F \subseteq D$ for a connected component $D$ of $L_p$, which would imply that $F$ is contained in a connected component of $L_{p}^{*}$. We claim that $F$ intersects at most one $0$-component of $L_p^*$. To show this, suppose by way of contradiction, that  $F$ intersects at least two $0$-components of $L_{p}^{*}$, which we denote by $D^{*}_i$ and $D^{*}_j$. From the definition of a $0$-component in $L_{p}^{*}$, it follows that $D^{*}_i$ and $D^{*}_j$ are also different $0$-components in $L_p$. Since $p$ is a good assignment and Rule 3 is no longer applicable, we have that $p(F)$ should be equal to $0$, which is a contradiction. Therefore, $F$ can intersect at most one $0$-component of $L_{p}^{*}$. From the construction of $L_{p}^{*}$, observe that $F$ cannot intersect more than one non $0$-component of $L_{p}^{*}$. Hence, we conclude that if $F \not \subseteq D^{*}$, for some connected component $D^{*}$ of $L_{p}^{*}$, then $F$ intersects exactly one $0$-component of $L_{p}^{*}$, say $D_{1}^{*}$, and exactly one non $0$-component of $L_{p}^{*}$, say $D_{2}^{*}$. That is, $F = (F \cap D_{1}^{*}) \cup (F \cap D_{2}^{*})$. This concludes the proof. 
\end{proof}

For any $0$-connected component \dstar{} of \lstarp, let $E_{{D}^{*}} = \{F : F \subseteq {D}^{*}\}$. For any non $0$-connected component \dstar of \lstarp, let $E_{{D}^{*}} = \{F : F \cap {D}^{*} \neq \emptyset\}$.

\begin{lemma} \label{claim:hp:segregatehyperedges}
$E(H) = \biguplus\limits_{{D}^{*} \in {L}^{*}_p} E_{{D}^{*}}$.
\end{lemma}
\begin{proof}
The lemma follows from the definition of $E_{D^{*}}$ and Lemma \ref{claim:hp:disjointhyperedges}.
\end{proof}

%For each connected component \dstar{} of \lstarp{} and for each $i \in [k]_0$, let us define a coloring function, $\Phi[{D}^{*},i] : {D}^{*} \to [k]_0$ as follows. For $i=0$, ${{\Upsilon}^{0}}_{{D}^{*}} ({D}^{*})=0$ and for each $i \in [k]$, $\Phi[{D}^{*},i] : {D}^{*} \to [k]_0$ such that, for each \emph{constituent} component $D$ of \dstar, $\Phi[{D}^{*},i] (D)= l_d$, if $L(D) \neq \emptyset$ where $L(D) = \{l_d\}$  and $\Phi[{D}^{*},i] (D)= i$, if $L(D) = \emptyset$.

For each connected component \dstar{} of \lstarp{} and for each $i \in [k]_0$, let us define the coloring function $\Phi[{D}^{*},i] : {D}^{*} \to [k]_0$ as follows. First, if $i=0$, then define $\Phi[{D}^{*},0]({D}^{*}) =0$, that is, for all $v\in D^*$, define $\Phi[{D}^{*},0]=0$. Otherwise, if $i \in [k]$, then for each \emph{constituent} component $D$ of \dstar\ such that $L(D) = \{l_d\}$, define $\Phi[{D}^{*},i](D)= l_d$, and for each constituent component $D$ of \dstar\ such that $L(D)=\emptyset$, define $\Phi[{D}^{*},i](D)= i$.

We now prove that for any connected component \dstar{} of \lstarp, there exists $i \in [k]_0$ such that our above definition of $\Phi[{D}^{*},i]$ precisely captures the way the hypothetical witnessing coloring $\Upsilon$ colors $D^*$.

\begin{lemma}\label{claim:hp:colorings}
For any connected component \dstar{} of \lstarp, there exists $i \in [k]_0$ such that ${\Upsilon}|_{{D}^{*}} =  \Phi[{D}^{*},i]$.
\end{lemma}

\begin{proof}
Let $D_1, \ldots, D_r$ be the constituent components of $D^{*}$. First assume that $r =1$, that is, there is only one constituent component of ${D}^{*}$. From Lemma \ref{claim:hp:prop1}, since $D^{*} = D_1$, we have that $\Upsilon(D^{*}) = i$ for some $i \in [k]_0$. Hence, if $L(D) = \emptyset$, then ${\Upsilon}|_{D^{*}} = \Phi[D^{*},i]$. Otherwise, let $L(D) = \{l_d\}$. In this case, from Lemma \ref{claim:hp:ld0}, either $ \Upsilon(D^{*}) = l_d$ or $\Upsilon(D^{*}) =0$. Thus, either ${\Upsilon}_{|D^{*}} = \Phi[D^{*},0]$ or ${\Upsilon}|_{D^{*}} = \Phi[D^{*},i]$ for any $i\in[k]$.

Now, we need prove the claim for the case where $r \geq 2$. Then, by Lemma \ref{claim:hp:emptylabel}, for any constituent component $D_i$ of $D^{*}$, we have that $L(D_i) \neq \emptyset$. For any $i \in [r]$, let $L(D_i) = \{l_{d_i}\}$. From the construction of $L_{p}^{*}$, for each $i \in [r]$, we have that $l_{d_i} \in [k]$, which in particular means that $l_{d_i} \neq 0$. From Lemma \ref{claim:hp:ld0}, either $\Upsilon(D_i)=0$ or $\Upsilon(D_i) = l_{d_i}$.
We aim to prove that either ${\Upsilon}|_{D^{*}} = \Phi[D^{*}, 0]$ or ${\Upsilon}|_{D^{*}} = \Phi[D^{*}, i]$ for any $i \in [k]$. In other words, we next show that either all constituent components of $D^{*}$ are colored $0$ by $\Upsilon$ or $\Upsilon$ colors each constituent component with the color represented by the label of that constituent component.
%Suppose not. 
To this end, let $\mathcal{D'}$ be the collection of all constituent components of $D^{*}$ such that for all $D_i \in \mathcal{D'}$, $\Upsilon(D_i) = l_{d_i}$ and let $\mathcal{D''}$ be the collection of all constituent components of $D^{*}$ such that for all $D_i \in \mathcal{D''}$, $\Upsilon(D_i) = 0$. We need to show that either $\mathcal{D'} = \emptyset$ or $\mathcal{D''} = \emptyset$. Suppose not, that is, $\mathcal{D'} \neq \emptyset$ and $\mathcal{D''} \neq \emptyset$. Then, there exist $D_1 \in \mathcal{D'}$ and $D_2 \in \mathcal{D''}$. Since $D^{*}$ is a connected component in $L_{p}^{*}$, there exists $F \in E(H)$ such that $F \cap D_1 \neq \emptyset$ and $F \cap D_2 \neq \emptyset$. Since $\Upsilon(D_1)= l_{d_1} \neq 0$, from Lemma~\ref{claim:hp:prop2}, we have that $F \in E_m$. Since $\Upsilon(D_2) = 0$, $F \cap D_2 \neq \emptyset$ and $F \in E_m$, we deduce that $D$ contains a vertex from the set $B$. Thus, by Lemma \ref{claim:hp:bigvertices}, we have that $D_2$ is a $0$-component of $L_{p}^{*}$. Since no constituent component of $D^{*}$ (which contains at least two constituent components) can be a $0$-component from the construction of $L_{p}^{*}$, we have reached a contradiction.
\end{proof}

\subsection{Dynamic Programming}

For the sake of clarity of presentation, for every hyperedge $F \in E_{D^{*}}$, we denote $h_F(\mu',l_1',l_2') = \allowbreak \bigvee_{i \in [k]_0} \allowbreak f_F({\Phi[{D}^{*},i]}|_{F}, \mu',l_1',l_2')$.

Moreover, we let $D_{1}^{*}, \ldots, D_{y}^{*}$ denote the connected components of \lstarp. For each $i \in [y]$, denote $E_{D^{*}_{i}} = \{F_{i,1}, \ldots, F_{i,{z}_{i}}\}$ (recall that $E_{D^{*}_{i}}$ was defined in Appendix \ref{sec:supergraph}). 
For all $i \in [y]$, $\mu'\in[\mu]$, $l_i'\in[l_1]_0$ and $l_2'\in[l_2]$,
%\muin, \lfin{} and \lsin, 
define $\mathcal{H}[i, \mu',l_1',l_2']$ as follows.

$$\mathcal{H}[i,\mu',l_1',l_2'] = \bigvee\limits_{\substack{(\mu^{j})_{j \in [z_i]} \\ (l_{1}^{j})_{j \in [z_i]}  \\ (l_{2}^{j})_{j \in [z_i]}}} \bigwedge\limits_{j \in [z_i]} h_{{F}_{i,j}} ({\mu}^j, {l_1}^{j}, {l_2}^j),$$
where $\mu' = \sum\limits_{j \in [z_i]} {\mu}^j $, $\sum\limits_{j \in [z_i]} {l_1^j} \leq l_1'$, $\sum\limits_{j \in [z_i]} {l_2}^{j} \leq l_2'$, and each ${\mu}^{j}$, $l_1^{j}$, $l_2^{j}$ is a non-negative integer.

\begin{lemma}\label{claim:hp:finalanshp}
$\displaystyle{\text{{\ttfamily aHP}}[\mu,l_1,l_2] = \bigvee\limits_{\substack{(\mu^{j})_{j \in [z_i]} \\ (l_{1}^{j})_{ j \in [z_i]} \\ (l_{2}^{j})_{j \in [z_i]} }} \bigwedge\limits_{i \in [y]} \mathcal{H}[i, {\mu}^j, l_1^{j}, l_2^j]}$,\\ 
where $\mu = \sum\limits_{j \in [z_i]} {\mu}^j $, $\sum\limits_{j \in [z_i]} l_1^{j} \leq l_1$, $\sum\limits_{j \in [z_i]} {l_2}^{j} \leq l_2$, and each ${\mu}^{j}$, $l_1^{j}$, $l_2^{j}$ is a non-negative integer.
\end{lemma}
\begin{proof}
By Lemma \ref{claim:hp:segregatehyperedges}, we have the following equality.
%that for any $0 \leq \mu \leq b$, $0 \leq l_1 \leq k_1$, $0 \leq l_2 \leq k_2$,

 $$\text{{\ttfamily aHP}}[\mu,l_1,l_2] = \bigvee_{\substack{\Upsilon : V(H) \to [k]_0 \\ (\mu^{F})_{F \in E(H)} \\ (l_{1}^{F})_{F \in E(H)} \\ (l_{2}^{F})_{F \in E(H)}}} \bigwedge\limits_{\substack{i\in[y]}} \bigwedge\limits_{ F \in  E_{D^{*}_i} }  f_F(\Upsilon|_{F}, {\mu}^{F}, l_{1}^{F}, l_{2}^{F}),$$ 
 where $\mu = \sum\limits_{F \in E(H)} {\mu}^{F}$, $\sum\limits_{F \in E(H)} l_{1}^{F} \leq l_1$, $\sum\limits_{F \in E(H)} l_{2}^{F} \leq l_2$, and for all $F\in F(H)$, ${\mu}^{F}$, $l_{1}^{F}$ and $l_{2}^{F}$ are non-negative integers.
 
From Lemma \ref{claim:hp:colorings}, we thus further have the following equality.
%that for any $0 \leq \mu \leq b$, $0 \leq l_1 \leq k_1$, $0 \leq l_2 \leq k_2$,
 
  $$\text{{\ttfamily aHP}}[\mu,l_1,l_2] = \bigvee_{\substack{(\mu^{F})_{F \in E(H)} \\ (l_{1}^{F})_{F \in E(H)} \\ (l_{2}^{F})_{F \in E(H)}}} \bigwedge\limits_{\substack{i\in[y]}} \bigwedge\limits_{ F \in  E_{D^{*}_i}}  h_F(\mu^F,l_1^F,l_2^F),$$ 
 where $\mu = \sum\limits_{F \in E(H)} {\mu}^{F}$, $\sum\limits_{F \in E(H)} l_{1}^{F} \leq l_1$, $\sum\limits_{F \in E(H)} l_{2}^{F} \leq l_2$, and for all $F\in F(H)$, ${\mu}^{F}$, $l_{1}^{F}$ and $l_{2}^{F}$ are non-negative integers.
  
Hence, by the definition of $\mathcal{H}[i, \mu',l_1',l_2']$, we conclude that the equation in the statement of the lemma is correct.
 
\end{proof}

\begin{lemma}\label{claim:hp:computeH}
Suppose that for all $i\in[y]$, $F \in E_{D^{*}_i}$, ${\mu}' \leq \mu$, $l_{1}' \leq l_1$ and $l_{2}' \leq l_2$, it holds that $h_F({\mu}', l_{1}', l_{2}')$ is computable in time $\tau$. Then, for all $i \in [y]$, ${\mu}' \leq \mu$, $l_{1}' \leq l_1$ and $l_{2}' \leq l_2$, it holds that $\mathcal{H}[i,\mu',l_1',l_2']$ can be computed in time $\OO(\tau \cdot z_i \cdot b^2 \cdot k_1^3 \cdot k_2^3)$. 
\end{lemma}
\begin{proof}
Arbitrarily choose some $i \in [y]$, ${\mu}^* \leq \mu$, $l_{1}^* \leq l_1$ and $l_{2}^* \leq l_2$. Under the given supposition, we would show that  $\mathcal{H}[i,\mu^*,l_1^*,l_2^*]$ can be computed in time $\OO(\tau \cdot z_i \cdot b^2 \cdot k_1^3 \cdot k_2^3)$. To this end, for all $c \in [z_i]$, ${\mu}' \leq \mu$, $l_{1}' \leq l_1$ and $l_{2}' \leq l_2$, define ${\mathcal{H}}_c[i, \mu',l_1',l_2']$ as follows.

$${\mathcal{H}}_c[i,\mu',l_1',l_2] = \bigvee\limits_{\substack{(\mu^{j})_{j \in [c]} \\ (l_{1}^{j})_{j \in [c]}  \\ (l_{2}^{j})_{j \in [c]} }} \bigwedge\limits_{j \in [c]} h_{{F}_{i,j}} ({\mu}^j, l_1^{j}, l_2^j),$$
where $\mu' = \sum\limits_{j \in [c]} {\mu}^j $, $\sum\limits_{j \in [c]} l_{1}^{j} \leq l_1'$, $\sum\limits_{j \in [c]} l_{2}^{j} \leq l_2'$ and each ${\mu}^{j}$, $l_{1}^{j}$, $l_{2}^j$ is a non-negative integer.

Now, for all ${\mu}' \leq \mu$, $l_{1}' \leq l_1$ and $l_{2}' \leq l_2$, $\mathcal{H}_{z_i}[i,\mu',l_1',l_2']$ can be computed in time $\OO(\tau \cdot z_i \cdot b^2 \cdot k_1^2 \cdot k_2^2)$ using the following recurrences.

$${\mathcal{H}}_1[i,\mu',l_1',l_2'] = h_{F_{i,1}}(\mu',l_1',l_2').$$

\noindent For all $c \in \{2, \ldots, z_i\}$, $i\in[y]$, ${\mu}' \leq \mu$, $l_{1}' \leq l_1$ and $l_{2}' \leq l_2$,
$$ {\mathcal{H}}_c[i,\mu',l_1',l_2'] = \bigvee\limits_{\substack{\mu' = {\mu}^1 + {\mu}^2\\  l_1' \geq l_{1}^{1} + l_{2}^{2}  \\  l_2' \geq l_{2}^{1} + l_{2}^{2}}}    {\mathcal{H}}_{c-1}[i,{\mu}^{1},l_1^{1},l_2^{1}]  \wedge h_{F_{i,c}} ({\mu}^2, l_1^2 , l_2^2).$$

Observe that $\mathcal{H}[i,\mu^*,l_1^*,l_2^*]={\mathcal{H}}_{z_i}[i,\mu^*,l_1^*,l_2^*]$. This concludes the proof.
\end{proof}

\begin{lemma}\label{claim:hp:computeAns}
Suppose that for all $i \in [y]$, ${\mu}' \leq \mu$, $l_{1}' \leq l_1$ and $l_{2}' \leq l_2$, $\mathcal{H}[i,{\mu}',l_{1}',l_{2}']$ can be computed in time $\psi$. Then, {\ttfamily aHP}$[\mu, l_1,l_2]$ can be computed in time $\OO(\psi \cdot y \cdot b^2 \cdot k_1^3 \cdot k_2^3)$. 
\end{lemma}
\begin{proof}
For all $i \in [y]$, ${\mu}' \leq \mu$, $l_{1}' \leq l_1$ and $l_{2}' \leq l_2$, let us first define ${\text{{\ttfamily aHP}}}_i[\mu',l_1',l_2']$ as follows.
$${\text{{\ttfamily aHP}}}_i[\mu',l_1',l_2'] = \bigvee\limits_{\substack{(\mu^{j})_{j \in [i]} \\ (l_{1}^{j})_{j \in [i]}  \\ (l_{2}^{j})_{j \in [i]} }} \bigwedge\limits_{j\in [i]} \mathcal{H}[j, {\mu}^j, l_{1}^{j}, l_{2}^j],$$
where $\mu' = \sum\limits_{j \in [c]} {\mu}^j $, $\sum\limits_{j \in [c]} l_{1}^{j} \leq l_1'$, $\sum\limits_{j \in [c]} l_{2}^{j} \leq l_2'$ and each ${\mu}^{j}$, $l_{1}^{j}$, $l_{2}^j$ is a non-negative integer.

Observe that ${\text{{\ttfamily aHP}}}[\mu,l_1,l_2]={\text{{\ttfamily aHP}}}_y[\mu,l_1,l_2]$ from Lemma \ref{claim:hp:finalanshp}.
Now, we can compute $\text{{\ttfamily aHP}}_y[\mu, l_1,l_2]$ using the following recurrences in time $\OO(\psi \cdot y \cdot b^2 \cdot k_1^2 \cdot k_2^2)$. 

$$\text{{\ttfamily aHP}}_1[\mu',l_1',l_2'] = \mathcal{H}[1, \mu', l_1',l_2'].$$

\noindent For all $i \in \{2, \ldots, y\}$, ${\mu}' \leq \mu$, $l_{1}' \leq l_1$ and $l_{2}' \leq l_2$, $$\text{{\ttfamily aHP}}_i[\mu',l_1',l_2']= \bigvee\limits_{\substack{\mu'={\mu}^1 + {\mu}^2 \\ l_1' \geq l_{1}^1 + l_{1}^2 \\ l_2' \geq l_{2}^1 + l_{2}^2}} {\text{{\ttfamily aHP}}}_{i-1} [{\mu}^1, l_{1}^1, l_{2}^1] \wedge \mathcal{H}[i, {\mu}^2 , l_{1}^2, l_{2}^2 ].$$

This concludes the proof of the lemma.
\end{proof}

Clearly, for all $i\in[y]$, $F \in E_{D^{*}_i}$, ${\mu}' \leq \mu$, $l_{1}' \leq l_1$ and $l_{2}' \leq l_2$, it holds that $h_F({\mu}', l_{1}', l_{2}')$ is computable in polynomial time. Thus, by Lemmas \ref{claim:hp:computeH} and \ref{claim:hp:computeAns}, we conclude the proof of Theorem \ref{thm:hp}.

\bibliography{bjp} {}
\bibliographystyle{plain}
\end{document}